%% file: TMC_final_version.tex
\documentclass[journal]{IEEEtran}

\makeatletter
\def\ps@headings{%
\def\@oddhead{\mbox{}\scriptsize\rightmark \hfil \thepage}%
\def\@evenhead{\scriptsize\thepage \hfil \leftmark\mbox{}}%
\def\@oddfoot{}%
\def\@evenfoot{}}
\makeatother
\usepackage{multirow}
\usepackage{xcolor}
\usepackage{graphicx,epstopdf}
\usepackage{amsmath,amsthm,amssymb,amsfonts}
\usepackage{tabularx,pdfcomment}

\usepackage{subfigure}
\usepackage{hyperref}
\usepackage{url}
\usepackage{cases}
\usepackage{bbm}
\usepackage{bm}

\usepackage{algorithmic}
\usepackage{algorithm}

\usepackage[square, comma, sort&compress, numbers]{natbib}

\theoremstyle{plain} \newtheorem{theorem}{Theorem}
\theoremstyle{plain} \newtheorem{definition}{Definition}
\theoremstyle{plain} \newtheorem{lemma}{Lemma}
\theoremstyle{plain} \newtheorem{remark}{Remark}
\theoremstyle{plain}

\newcommand{\smallfont}{\fontsize{10pt}{10pt}\selectfont}

\newcommand{\mc}[1]{\mathcal{#1}}
\newcommand{\mb}[1]{\mathbb{#1}}
\newcommand{\V}[1]{{\bm{\mathbf{\MakeLowercase{#1}}}}} 

\newcommand{\M}[1]{{\bm{\mathbf{\MakeUppercase{#1}}}}} 

\begin{document}

\title{Adaptive Sampling of RF Fingerprints for Fine-grained Indoor Localization}

\author{Xiao-Yang Liu, Shuchin Aeron, {\em Senior Member, IEEE}, Vaneet Aggarwal, {\em Senior Member, IEEE}, \\
Xiaodong Wang, {\em Fellow, IEEE}, and Min-You Wu, {\em Senior Member, IEEE}
\thanks{ X.~Liu and X.~Wang are with the Department of Electrical Engineering, Columbia University.
  S.~Aeron is with the Department of Electrical and Computer Engineering, Tufts University.
 V.~Aggarwal is with the School of Industrial Engineering, Purdue University.
 X.~Liu and M.~Wu are with the Department of Computer Science and Engineering, Shanghai Jiao Tong University.}
}

\maketitle

\begin{abstract}

  Indoor localization is a supporting technology for a broadening range of pervasive wireless applications. One promising approach is to locate users with radio frequency fingerprints. However, its wide adoption in real-world systems is challenged by the time- and manpower-consuming site survey process, which builds a fingerprint database {\em a priori} for localization. To address this problem, we visualize the 3-D RF fingerprint data as a function of locations (x-y) and indices of access points (fingerprint), as a \emph{tensor} and use tensor algebraic methods for an \emph{adaptive} tubal-sampling of this fingerprint space. In particular using a recently proposed tensor algebraic framework in \cite{Kilmer2013SIAM} we capture the complexity of the fingerprint space as a low-dimensional tensor-column space. In this formulation the proposed scheme exploits adaptivity to identify reference points which are highly informative for learning this low-dimensional space. Further, under certain incoherency conditions we prove that the proposed scheme achieves bounded recovery error and near-optimal sampling complexity. In contrast to several existing work that rely on random sampling, this paper shows that adaptivity in sampling can lead to significant improvements in localization accuracy. The approach is validated on both data generated by the ray-tracing
indoor model which accounts for the floor plan and the impact of walls and the real world data. Simulation results show that, while maintaining the same localization accuracy of existing approaches, the amount of samples can be cut down by $71\%$ for the high SNR case and $55\% $ for the low SNR case.
\end{abstract}

\begin{IEEEkeywords}Fine-grained indoor localization, RF fingerprint, adaptive sampling, tubal-sampling, tensor completion.
\end{IEEEkeywords}

\input{Sections/Introduction}

\input{Sections/Model}

\input{Sections/Baseline}

\input{Sections/OurApproach}

\input{Sections/Simulation}

\section{Conclusions}

  In this paper, an adaptive sampling approach is proposed to relieve the site survey burden of fingerprint-based indoor localization. It cuts down the sampling budget by $71\%$ for the high SNR case and $55\% $ for the low SNR case while maintaining a similar localization error performance of widely used localization schemes (KNN, the kernel approach and SVM). The performance gain comes from the basic observation that the RF fingerprints are highly correlated across space and across access points and it is possible to adaptively locate more informative entries. Since RF fingerprint calibration is tubal-sampling, we used low tubal-rank tensors to model RF fingerprints instead of low CP-rank, and proposed an algorithm for reconstruct the fingerprint database which adaptively sampling a subset of the reference points and then performing tensor completion. We show that the proposed scheme achieves near-optimal sampling complexity. Our results solve a major challenge of fingerprint-based indoor localization and we advocate its wide adoption in real-world systems.

\section*{Acknowledgements}
The authors would like to thank Rittwik Jana, Sarat Puthenpura, and N. K. Shankaranayanan at AT\&T Labs-Research, for helpful discussions on propagation models and localization used in this paper. Xiao-Yang Liu gratefully acknowledges financial support from China Scholarship Council.

This work was supported in part by the NSF CCF award 1319653, and NSFC  projects - 61373155, 91438121, and 61373156.

\input{Sections/Bib}

\input{Sections/Proof}

\begin{IEEEbiography}[{\includegraphics[width=1in,height=1.25in,clip,keepaspectratio]{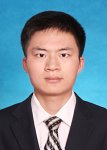}}]{Xiao-Yang Liu} received his B.Eng. degree in
computer science from Huazhong University
of Science and Technology, China, in 2010.
He is currently working toward the PhD degree
in the Department of Computer Science
and Engineer in Shanghai Jiao Tong University.
His research interests include wireless
communication, distributed
systems, big data analysis, cyber security, and sparse optimization.
\end{IEEEbiography}

\begin{IEEEbiography}[{\includegraphics[width=1in,height=1.25in,clip,keepaspectratio]{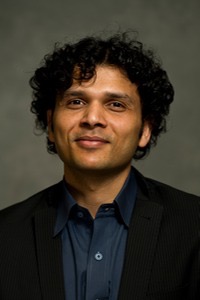}}]{Shuchin Aeron}(M'08 -  SM'15) is currently an assistant professor
in the department of ECE at Tufts University.
He received his Ph.D. in ECE from Boston University
in 2009. From 2009 to 2011 he was a post-doctoral
research fellow at Schlumberger-Doll Research
(SDR), Cambridge, MA where he worked on signal
processing answer products for borehole acoustics.
He has several patents and the proposed workflows
are currently implemented in logging while drilling
tools. He is the recipient of the best thesis award from
both the college of engineering and the department of
ECE in 2009. He received the Center of Information and Systems Engineering
(CISE) award from Boston University in 2006 and a Schlumberger-Doll Research
grant in 2007. His main research interest lie in statistical signal processing,
information theory, and optimal sampling and recovery.
\end{IEEEbiography}

\begin{IEEEbiography}[{\includegraphics[width=1in,height=1.25in,clip,keepaspectratio]{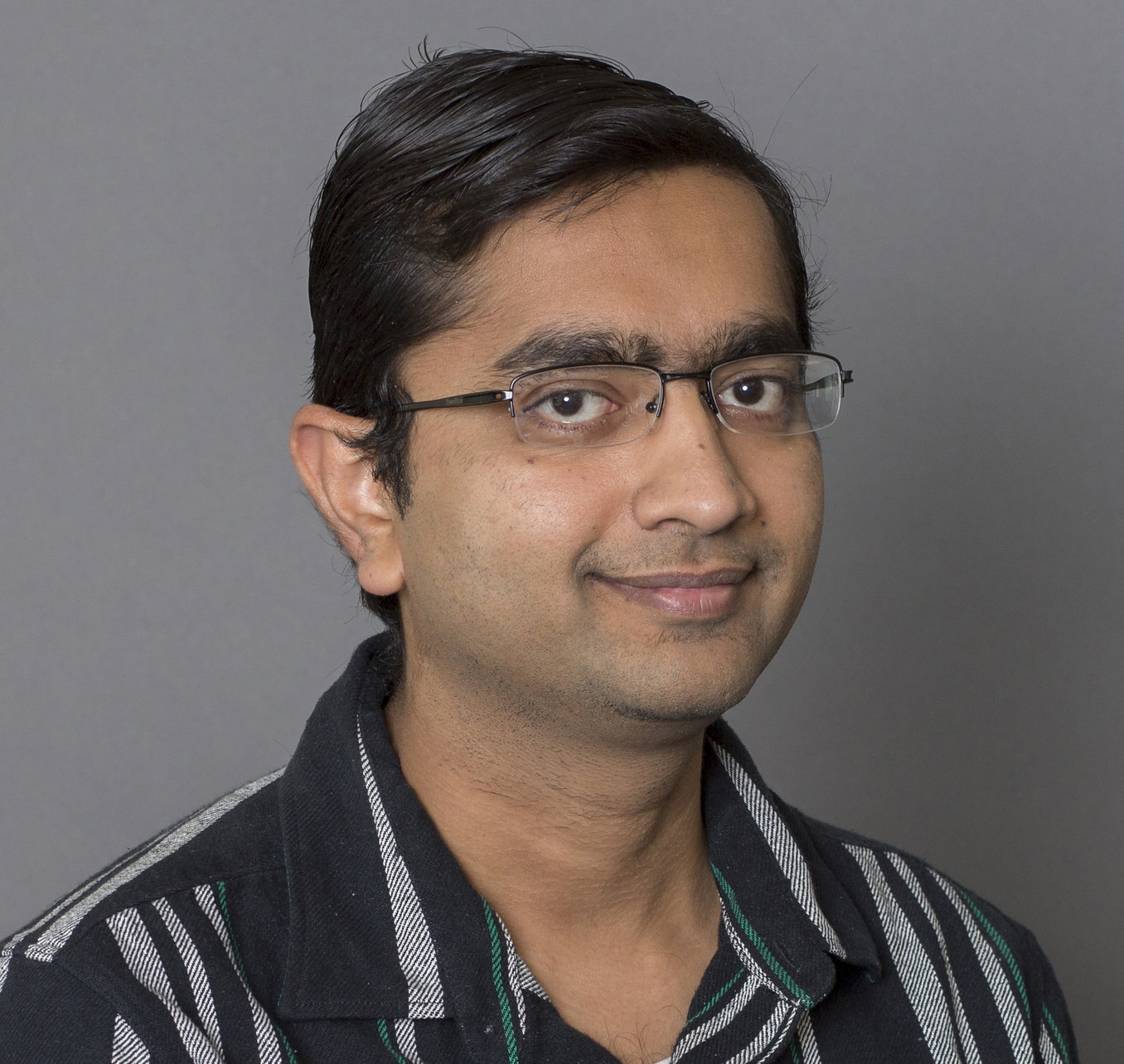}}]{Vaneet Aggarwal}(S'08 - M'11 - SM'15) received the B.Tech. degree in 2005 from the Indian Institute of Technology, Kanpur, India, and the M.A. and Ph.D. degrees in 2007 and 2010, respectively from Princeton University, Princeton, NJ, USA, all in Electrical Engineering.

He is currently an Assistant Professor at Purdue University, West Lafayette, IN. Prior to this, he was a Senior Member of Technical Staff Research at AT\&T Labs-Research, NJ, and an Adjunct Assistant Professor at Columbia University, NY. His research interests are in applications of information and coding theory to wireless systems, machine learning, and distributed storage systems. Dr. Aggarwal was the recipient of Princeton University's Porter Ogden Jacobus Honorific Fellowship in 2009.
\end{IEEEbiography}

\begin{IEEEbiography}[{\includegraphics[width=1in,height=1.25in,clip,keepaspectratio]{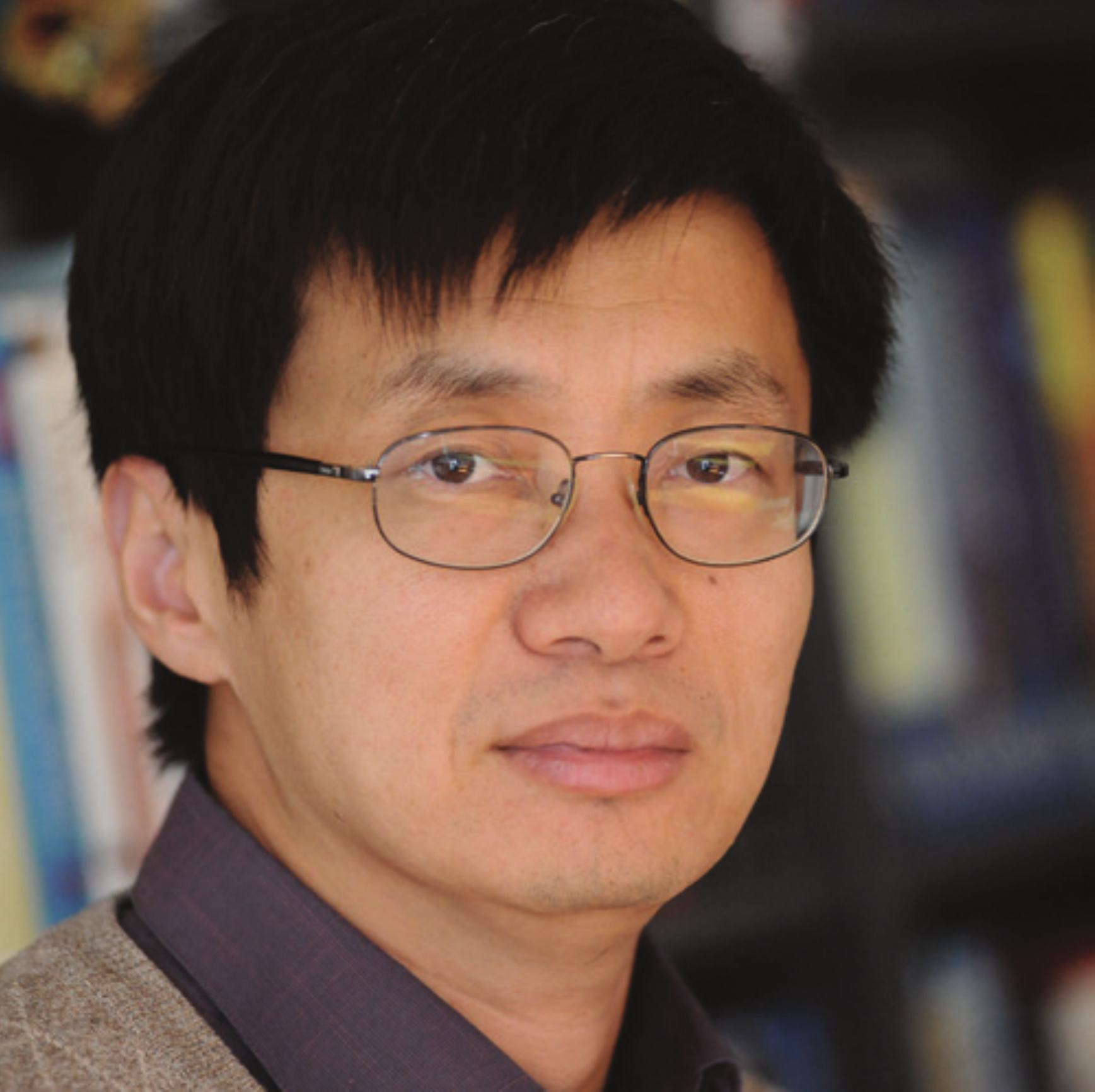}}]{Xiaodong Wang}(S'98-M'98-SM'04-F'08) received the Ph.D degree in Electrical Engineering from Princeton University. He is a Professor of  Electrical Engineering at Columbia University in New York. Dr. Wang's research interests fall in the general areas of computing, signal processing and communications, and has published extensively in these areas. Among his publications is a book entitled ``Wireless Communication Systems: Advanced Techniques for Signal Reception'', published by Prentice Hall in 2003.  His current research interests include wireless communications,   statistical signal processing, and genomic signal processing. Dr. Wang received the 1999 NSF CAREER Award, the 2001 IEEE Communications Society and Information Theory Society Joint Paper Award, and the 2011 IEEE Communication Society Award for Outstanding Paper on New Communication Topics. He has served  as an Associate Editor for the {\em IEEE Transactions on Communications}, the {\em IEEE Transactions on Wireless Communications}, the {\em IEEE Transactions on Signal Processing}, and the {\em IEEE Transactions on Information Theory}. He is a Fellow of the IEEE and listed as an ISI Highly-cited Author.
\end{IEEEbiography}

\begin{IEEEbiography}[{\includegraphics[width=1in,height=1.25in,clip,keepaspectratio]{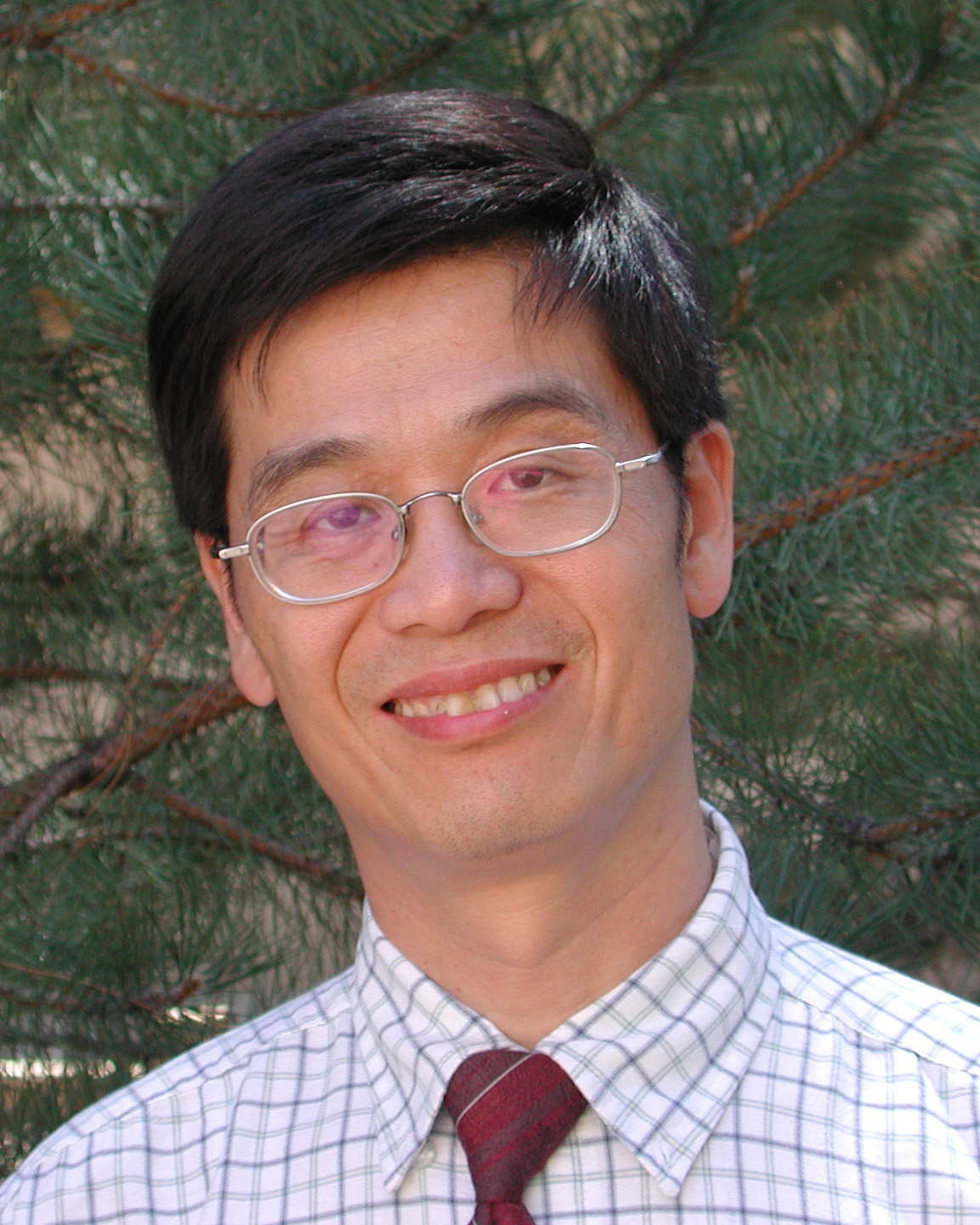}}]{Min-You Wu} (S'84-M'85-SM'96) received the M.S. degree from the Graduate School of Academia Sinica, Beijing, China, in 1981 and the Ph.D. degree from Santa Clara University, CA in 1984. He is a professor in the Department of Computer Science and Engineering at Shanghai Jiao Tong University and a research professor with the University of New Mexico. He serves as the Chief Scientist at Grid Center of Shanghai Jiao Tong University. His research interests include wireless networks, sensor networks, multimedia networking, parallel and distributed systems, and compilers.
\end{IEEEbiography}

\end{document}

%% file: Sections/Introduction.tex
\section{Introduction}

  The availability of real-time high-accuracy location-awareness in indoor environments is a key enabler for a wide range of pervasive wireless applications, {such as pervasive healthcare \cite{Kang2010TMC} and smart-space interactive environments \cite{Roy2012TMC,Dina2014NSDI}}. More recently, location-based services have been used in airports, shopping malls, supermarkets, stadiums, office buildings, and homes \cite{LBS,Chen2014MobiCom}. The economic impact of the indoor localization market is forecasted by ABI Research to reach $\$4$ billion in 2018 \cite{Market}.


  Indoor localization systems generally follow three approaches, viz., cell-based approach, model-based approach, and fingerprint-based approach. For cell-based approach \cite{Voronois2003,Barabasi2005}, a user's location is given by the access point to which it is connected and the localization error depends on the communication range and the distances between the access points. {This approach, however, can only provide coarse grained localization.} On the other hand model-based schemes \cite{Poor2005SPG,Win2011} exploit angle of arrival (AoA), time or time difference of arrival (ToA, TDoA), or received signal strength (RSS) from access points. However, such model-based approach is essentially limited by the following factors: 1) low transmission power, 2) high attenuation caused by walls and furnitures, 3) complicated surface reflections, and 4) unpredictable dynamics such as disturbance by human movements.


We will focus on the fingerprint-based approach, which can be classified into three categories, namely, RF (radio frequency) fingerprint-based, non-RF fingerprint-based, and cross-technology-based. For example, RADAR \cite{RADAR2000} is considered as the first RF fingerprint-based system; Google Indoor Map \cite{GoogleMap} and WiFiSLAM \cite{WiFiSLAM} are two widely used industrial apps, while the former uses RF fingerprints, the latter incorporates user trajectories, inertial information and accelerometer; IndoorAtlas \cite{IndoorAtlas} utilizes the magnetic field map; SurroundSense \cite{Azizyan2009MobiCom} exploits ambient attributes (sound, light, etc) including RF fingerprints. The advantages of RF fingerprint-based approach include: 1) It is a passive approach that exploits WiFi access points already in most buildings thus needs no extra infrastructure deployment; 2) The RSS values are provided by off-the-shelf WiFi- or Zigbee-compatible devices; 3) The flourishing smartphone market indicates its upcoming wide use; 4) It assumes no radio propagation model thus is more practical than the model-based approach.

{In this paper, we consider the RF fingerprint-based approach. It is a two-phase approach, a training phase (site survey) and an operating phase (location query). The fingerprints obtained are averaged over time to counteract the effects of fading, and thus usually change in a building with change of walls and/or
furniture. Assuming that the fingerprint data remains statistically stable, i.e., the mean RSS signal strengths from the access points don't change rapidly and that the variance in the change is small and can be captured as a small additive noise}. {In the training phase}, an engineer uses a smartphone to record RF fingerprints within a region of interest. Thereafter, a fingerprint database is built up at the server, in which each fingerprint is associated with the corresponding reference point. In the operating phase, a user submits a location query with her current  fingerprint, then the server responds by matching the query fingerprint with candidate reference points in the database.

{However, wide adoption of RF fingerprint-based approach is} challenged by the time- and manpower-consuming site survey process as the engineer needs to sample a large number of reference points when
the designer expects a significant changes of fingerprints. For example, a region of $100$m $\times$ $100$m covers $10^4$ reference points with grid size $1$m $\times$ $1$m. If a reference point takes about $10$ seconds (including moving to this point and measuring a stable fingerprint), then the site survey process takes about $27.8$ hours. Setting the grid size to a finer granularity will exacerbate this problem. In the end of 2014, Google Indoor Map 6.0 provides indoor localization and navigation only at some of the largest retailers, airports and transit stations in the U.S. and Japan \cite{GoogleMap}, while its expansion is constrained by limited amount of fingerprint data of building interiors.

  Recently, there are many works aiming to relieve the site survey burden, which we broadly classify into correlation-aware approach, crowdsourcing-based approach, and sparsity-aware approach. The correlation-aware approach leverages the fact that the fingerprints at nearby reference points are spatially correlated. For example, \cite{Li2005IEE} utilized the krigging interpolation method while \cite{Kuo2011TMC} adopted the kernel functions (continuous and differentiable discriminant functions) and proposed an approach based on discriminant minimization search. Such schemes try to use linear or non-linear functions to model the correlations, which are essentially model-based approaches and face similar limiting factors as discussed before.

  Crowdsourcing-based approach removes the offline site survey and instead {incorporates users' online cooperation. For example, LiFS \cite{Yang2012MobiCom} leverages user reports to establish a one-to-one mapping between RF fingerprints and the digital floor plan map. However, the digital map is not always available. To overcome this Jigsaw \cite{Jigsaw2014MobiCom} exploited image processing techniques to reconstruct floor plans based on the reported pictures and Zee \cite{Rai2012MobiCom} leveraged the inertial sensors (e.g., accelerometer, compass, gyroscope) to track users as they traverse an indoor environment while simultaneously performing WiFi scans.} Since the crowdsourcing-based approach is fundamentally the application of machine learning techniques to large-scale data sets and will be effective only when users' reports are sufficient to cover the whole space of human activity, it requires a large number of users to cooperate with the server. Furthermore, it imposes high energy consumption on smartphones.


Recent methods for efficient RF fingerprinting \cite{Zhang2013WCNC,Hu2013,Milioris2014,Nikitaki1, Nikitaki2,Milioris} assume that the matrix of RSS values across channels/access points and locations is low rank and nuclear norm minimization \cite{Candes2010} under \emph{random or deterministic non-adaptive element-wise sampling} constraints is used for data completion. In \cite{Tsakatakis} the authors assume that the signals are fast varying, model the signal variation as a first order dynamical system, and give approaches for  dynamic matrix completion. Below we summarize the main contributions of our work and contrast them with the approaches taken so far.

\subsection{Summary of Contributions}
{In this paper our main technical contributions are as follows.}
{
\begin{itemize}
\setlength \itemsep{-3pt}
\item We model the RF fingerprint data as a low rank \emph{tensor} using the notion of tensor rank proposed in \cite{Kilmer2013SIAM,Shuchin}. This is in contrast to existing works \cite{Zhang2013WCNC,Hu2013,Milioris2014,Nikitaki1, Nikitaki2,Milioris} that assume a low rank matrix model. Furthermore, our algebraic framework is different from the traditional \emph{multilinear algebraic framework} for tensor decompositions \cite{Tensor2009SIAM} that has been considered so far in the literature for problems of completing multidimensional arrays \cite{Sujay,Oh,Yamada} with different notions for tensor rank.
\item We propose an \emph{adaptive} sampling strategy that leads to a dramatic improvement in localization accuracy for the same sample complexity. Our approach  adaptively samples a small subset of all reference points based on which, the fingerprint data is then reconstructed using tensor completion. In this context, a major difference from existing low-rank tensor completion \cite{Tensor2009SIAM,Tensor2011,Shuchin,Shuchin2015PAMI,Singh2013NIPS} lies in that {the sampling strategy} in our paper is a vector-wise sampling, while existing low-rank tensor completion deals with entry-wise sampling. This is because when measurements is done at a location, there
is data for all of the access points.
\item We also derive theoretical performance bounds and show that it is possible to recover data under weaker conditions than required for non-adaptive random sampling considered in \cite{Shuchin}. This is further evidenced by our numerical simulation on (i) a software model using concepts of ray-tracing and supervised learning with real data \cite{SRT,Ray_tracing_TWC2009}, which accounts for shadowing of walls, wave guiding effects in corridors due to multiple reflections, diffractions around vertical wedges, and (ii) real data.  We observe orders of magnitude improvements in the completion performance, see Fig. \ref{fig:recovery} and \ref{fig:real_data_recovery}, over existing methods.
\item Applying the completed data for localization also results in orders of magnitude performance improvement in localization accuracy (see Fig. \ref{fig:KNN}, \ref{fig:kernel} and \ref{fig:SVM}) over other competing methods. Therefore our approach of exploiting adaptivity is efficient in collecting measurements while maximizing the localization accuracy.
\end{itemize}}

The remainder of the paper is organized as follows. In Section II, we present the system model and the problem statement. Section III describes a random sampling approach as a baseline. Section IV provides details of our approach while the performance guarantees are given in Section V. Simulation results are presented in Section VI. Detailed proofs are given in the Appendix, and concluding remarks are made in Section VII.

%% file: Sections/Model.tex
\section{System Model and Problem Statement}
We begin by first outlining the notations, the algebraic model and preliminary results for third-order tensors taken from \cite{Kilmer2013SIAM,Shuchin,Shuchin2015PAMI}.

  A third-order tensor is represented by calligraphic letters, denoted as $\mathcal{T} \in \mathbb{R}^{N_1 \times N_2 \times N_3}$, and its $(i,j,k)$-th entry is $\mathcal{T}(i,j,k)$. A {\em tube} (or fiber) of a tensor is a 1-D section defined by fixing all indices but one, thus a tube is a vector. In this paper, we use tube $\mathcal{T}(i,j,:)$ to denote a fingerprint at reference point $(i,j)$. Similarly, a {\em slice} of a tensor is a 2-D section defined by fixing all but two indices. {\em frontal, lateral, horizontal slices} are denoted as $\mathcal{T}(:,:,k),\mathcal{T}(:,j,:),\mathcal{T}(i,:,:)$, respectively. We use lateral slice $\mathcal{T}(:,j,:)$ to denote an $N_1 \times 1 \times N_3$ fingerprint matrix for the $j$-th column of the grid map.

$\breve{\mathcal{T}}$ is a tensor obtained by taking the Fourier transform along the third mode of $\mathcal{T}$, i.e., $\breve{\mathcal{T}}(i,j,:) = \text{fft}(\mathcal{T}(i,j,:))$. In MATLAB notation, $\breve{\mathcal{T}} = \text{fft}(\mathcal{T},[~],3)$, and one can also compute $\mathcal{T}$ from $\breve{\mathcal{T}}$ via $\mathcal{T} = \text{ifft}(\breve{\mathcal{T}},[~],3)$.

  The transpose of tensor $\mathcal{T}$ is the $N_2 \times N_1 \times N_3$ tensor $\mathcal{T}^\top$ obtained by transposing each of the frontal slices and then reversing the order of transposed frontal slices $2$ through $N_3$, i.e., for $k=2,3,...,N_3$, $\mathcal{T}^\top(:,:,k)=(\mathcal{T}(:,:,N_3+2-k))^\top$ (the transpose of matrix $\mathcal{T}(:,:,N_3+2-k)$). For two tubes $a,b \in \mathbb{R}^{1 \times 1 \times N_3}$, $a \ast b$ denotes the \emph{circular convolution} between these two vectors. 

  The block diagonal matrix $blkdiag(\breve{\mathcal{T}})$ is defined by placing the frontal slices in the diagonal, i.e., $\text{diag}(\breve{\mathcal{T}}(:,:,1), \breve{\mathcal{T}}(:,:,2),...,\breve{\mathcal{T}}(:,:,N_3))$.

{The algebraic development in \cite{Kilmer2013SIAM} rests on defining a tensor-tensor product between two 3-D tensors, referred to as the t-product as defined below.}

  \begin{definition}
  \textbf{t-product}. The t-product $\mathcal{C} = \mathcal{A} \ast \mathcal{B}$ of $\mathcal{A} \in \mathbb{R}^{N_1 \times N_2 \times N_3}$ and $\mathcal{B} \in \mathbb{R}^{N_2 \times N_4 \times N_3}$ is a tensor of size $N_1 \times N_4 \times N_3$ whose $(i,j)$-th tube $\mathcal{C}(i,j,:)$ is given by
  $\mathcal{C}(i,j,:) = \sum\limits_{k=1}^{N_2} \mathcal{A}(i,k,:) \ast \mathcal{B}(k,j,:)$, 
  for $i=1,2,...,N_1$ and $j=1,2,...,N_4$.
  \end{definition}
Owing to the relation between the circular convolution and Discrete Fourier Transform, we note the following remark that is used throughout the paper.
   \begin{remark}\label{computation}
   For $\mathcal{A}\in \mathbb{R}^{N_1 \times N_2 \times N_3}$ and $\mathcal{B}\in \mathbb{R}^{N_2 \times N_4 \times N_3}$, we have
   $\mathcal{A} \ast \mathcal{B} = \mathcal{C} \Longleftrightarrow \text{ blkdiag}(\breve{\mathcal{A}})  ~ \text{ blkdiag}(\breve{\mathcal{B}}) =   \text{ blkdiag}(\breve{\mathcal{C}})$.
   \end{remark}

  A third-order tensor of size $N_1 \times N_2 \times N_3$ can be viewed as an $N_1 \times N_2$ matrix of tubes that are in the third-dimension. So the t-product of two tensors can be regarded as multiplication of two matrices, except that the multiplication of two numbers is replaced by the circular convolution of two tubes. {Further, this allows one to treat 3-D tensors as \emph{linear operators} over 2-D matrices as analyzed in \cite{Kilmer2013SIAM}. Using this perspective one can define a SVD type decomposition, referred to as the tensor-SVD or t-SVD \cite{}. To define the t-SVD we introduce a few definitions.}

  \begin{definition}
  \textbf{Identity tensor}.
  The identity tensor $\mathcal{I} \in \mathbb{R}^{N_1 \times N_1 \times N_3}$ is a tensor whose first frontal slice $\mathcal{I}(:,:,1)$ is the $N_1 \times N_1$ identity matrix and all other frontal slices are zero.
  \end{definition}


  \begin{definition}
  \textbf{Orthogonal tensor}. A tensor $\mathcal{Q} \in \mathbb{R}^{N_1 \times N_1 \times N_3}$ is orthogonal if it satisfies
  $\mathcal{Q}^\top \ast \mathcal{Q} =  \mathcal{Q} \ast \mathcal{Q}^\top = \mathcal{I}$.
  \end{definition}

  \begin{definition}\label{inverse}
  The inverse of a tensor $\mathcal{U} \in \mathbb{R}^{N_1 \times N_1 \times N_3}$ is written as $\mathcal{U}^{-1} \in \mathbb{R}^{N_1 \times N_1 \times N_3}$ and satisfies
  $\mathcal{U}^{-1} \ast \mathcal{U} = \mathcal{U} \ast \mathcal{U}^{-1} = \mathcal{I}$.
  \end{definition}

    \begin{definition}
  \textbf{f-diagonal tensor}.
  A tensor is called f-diagonal if each frontal slice of the tensor is a diagonal matrix, i.e., $\mathcal{T}(i,j,k)=0$ for $i \neq j, \forall k$.
  \end{definition}

{Using this definition one can obtain the t-SVD defined in the following result from \cite{Kilmer2013SIAM}. Please see Figure \ref{fig:t_SVD} for a graphical representation}.

  \begin{theorem}\label{tsvd}
   \textbf{t-SVD}. A tensor $\mathcal{T} \in \mathbb{R}^{N_1 \times N_2 \times N_3}$, can be decomposed as
   $\mathcal{T} = \mathcal{U} \ast \Theta \ast \mathcal{V}^\top$,
   where $\mathcal{U}$ and $\mathcal{V}$ are orthogonal tensors of sizes $N_1 \times N_1 \times N_3$ and $N_2 \times N_2 \times N_3$ respectively, i.e. $\mathcal{U}^\top \ast \mathcal{U} = \mathcal{I}$ and $\mathcal{V}^\top \ast \mathcal{V} = \mathcal{I}$ and $\Theta$ is a rectangular f-diagonal tensor of size $N_1 \times N_2 \times N_3$.
   \end{theorem}

   \begin{definition}\label{tubal_rank}
   \textbf{Tensor tubal-rank}. The tensor tubal-rank of a third-order tensor is the number of non-zero fibers of $\Theta$ in the t-SVD.
   \end{definition}

   \begin{figure}
   \centering
   \includegraphics[width=0.48\textwidth]{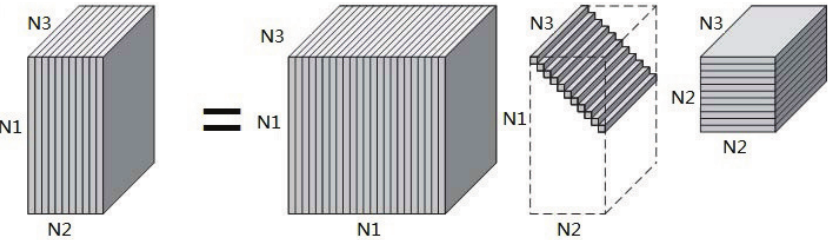}
   \vspace{-.1in}
   \caption{The t-SVD of an $N_1 \times N_2 \times N_3$ tensor.} \label{fig:t_SVD}
   \end{figure}
{In this framework, the principle of dimensionality reduction follows from the following result from \cite{Kilmer2013SIAM} \footnote{Note that $\Theta$ in t-SVD is organized in a decreasing order, i.e., $||\Theta(1,1,:)||_F \geq ||\Theta(2,2,:)||_F \geq ... $, which is implicitly defined in \cite{Shuchin} as the algorithm for computing t-SVD is based on matrix SVD. Therefore, the best rank-$r$ approximation of tensors is similar to PCA (principal component analysis).}.}
   \begin{lemma}\label{best_r_rank}
   \textbf{Best rank-$r$ approximation}.
   Let the t-SVD of $\mathcal{T} \in \mathbb{R}^{N_1 \times N_2 \times N_3}$ be given by $\mathcal{T} = \mathcal{U} \ast \Theta \ast \mathcal{V}^\top$ and for $r \leq \min(N_1,N_2)$ define $\mathcal{T}_r = \sum_{i=1}^{r} \mathcal{U}(:,i,:) \ast \Theta(i,i,:) \ast \mathcal{V}^\top(:,i,:)$, then
   $\mathcal{T}_r = \arg~\min\limits_{\overline{\mathcal{T}} \in \mathbb{T}}||\mathcal{T} - \overline{\mathcal{T}}||_F$,
   where $\mathbb{T} = \{ \Im = \mathcal{X} \ast \mathcal{Y} | \mathcal{X} \in \mathbb{R}^{N_1 \times r \times N_3}, \mathcal{Y} \in \mathbb{R}^{r \times N_2 \times N_3} \}$.
   \end{lemma}

{We now define the notion of tensor column space}. Under t-SVD in Definition \ref{tsvd}, a tensor column subspace of $\mathcal{T}$ is the space spanned by the lateral slices of $\mc{U}$ under the t-product, i.e., the set generated by $t$-linear combinations like so,
   $\text{t-span}(\mc{U}) = \{ \M{X} = \sum_{j = 1}^{r} \mc{U}(:,j,:) \ast \V{c}_j \in \mathbb{R}^{N_1 \times 1 \times N_3}\,\, , \V{c}_j \in \mb{R}^{N_3} \}$,
   where $r$ denotes the tensor tubal-rank. {We are now ready to present our main \emph{system model and assumptions}.}

\subsection{System Model}

  Suppose that the region of interest is a rectangle $\mathcal{R} \in \mathbb{R}^2$. Dividing $\mathcal{R}$ into an $N_1 \times N_2$ {\em grid map}, with each grid of the same size. The grid points are called {\em reference points}. Let $\mathcal{G}$ denote the grid map and $\mathcal{G}$ has $N_1 N_2$ reference points in total. Within $\mathcal{R}$, there are $N_3$ randomly deployed access points. We neither have access to or control over those access points, nor know their exact locations. The engineer uses a smartphone to measure the RSS values from these $N_3$ access points.
  We use a third-order tensor $\mathcal{T} \in \mathbb{R}^{N_1 \times N_2 \times N_3}$ to represent the RSS map of $\mathcal{G}$.
  Each reference point $(i,j) \in \mathcal{G}$ is associated with a received signal strength (RSS) vector $\mathcal{T}(i,j,:)$, called a {\em fingerprint}, where $\mathcal{T}(i,j,k)$ is the RSS value of the $k$-th access point. Note that the noise level is assumed to be equal to $-110$dBm, $\mathcal{T}(i,j,k)=-110$dBm \cite{Kuo2011TMC,Yang2012MobiCom} if the signal of the $k$-th access point cannot be detected. The fingerprint database stores the coordinates of all reference points and their corresponding RF fingerprints.

  We use a third-order tensor $\mathcal{T} \in \mathbb{R}^{N_1 \times N_2 \times N_3}$ to represent the RSS map of $\mathcal{G}$, and the t-SVD is $\mathcal{T} = \mathcal{U} \ast \Theta \ast \mathcal{V}^{T}$. 
  As mentioned in the Introduction, the RSS value is highly correlated across space for each access point and also across these $N_3$ access points. {We model this correlation by assuming that $\mathcal{T}$ has low tensor tubal-rank $r$, in the sense defined before using the t-SVD}.

  {We now outline the problem statement in detail. Note that as pointed out before, our sampling strategy will be to adaptively sample RF fingerprints, which correspond to the tensor fibers along the third dimension. See Figure \ref{fig:tubal_samp}.}

   \begin{figure}[t]
   \centering
   \includegraphics[width=0.3\textwidth]{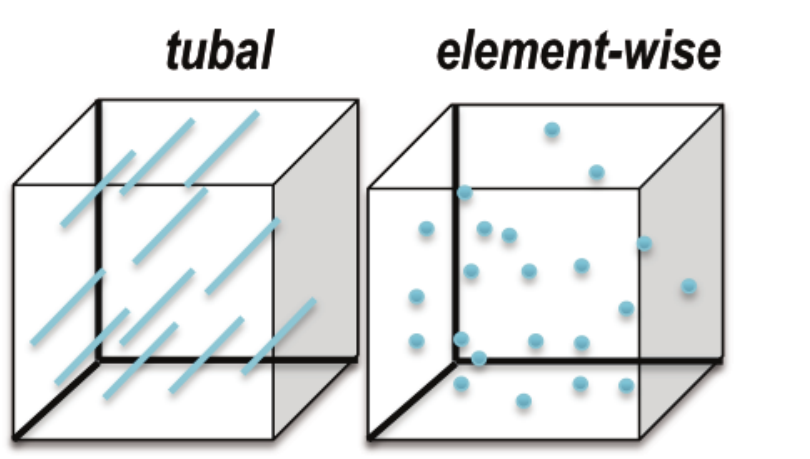}
   \caption{Figure contrasting tubal sampling with element-wise sampling for 3D tensors.} \label{fig:tubal_samp}
   \end{figure}

\subsection{Problem Statement}\label{problem_statement}

  Let $M < N_1 N_2$ denote the sampling budget, i.e., we are allowed to sample $M$ reference points. We model the site survey process as the following {\em partial observation model}:
  \begin{equation}\label{obs_model}
  \mathcal{Y} = \mathcal{P}_{\Omega}(\mathcal{T}) + \mathcal{P}_{\Omega}(\mathcal{N}),~~\Omega \subset \mathcal{G},
  \end{equation}
  where the $(i,j,k)$-th entry of $\mathcal{P}_{\Omega}(\mathcal{X})$ is equal to $\mathcal{X}(i,j,k)$ if $(i,j)\in \Omega$ and zero otherwise, $\Omega$ being a subset of the grid map $\mathcal{G}$ and of size $M$, and $\mathcal{N}$ is an $N_1 \times N_2 \times N_3$ tensor with i.i.d.  $N(0,\sigma^2)$ elements, representing the additive Gaussian noise.
   Since the engineer usually averages the recorded RSS values to get a stable fingerprint in the site survey process, while users want to get quick response from the server, the noise in the query process is much higher than that in the site survey samples $\mathcal{Y}$.

  To cut down the grid map survey burden, we measure the RSS values of a small subset of reference points and then estimate $\mathcal{T}$ from the samples $\mathcal{Y}$. There are two facts that can be exploited: the prior information that tensor $\mathcal{T}$ is low-tubal-rank, and the estimated $\hat{\mathcal{T}}$ should equal to $\mathcal{Y}$ on the set $\Omega$. Therefore, we estimate $\mathcal{T}$ by solving the following optimization problem:

  \begin{equation}\label{problem_formulation}
  \langle \hat{\mathcal{T}}, \hat{\Omega} \rangle = \arg \min_{\mathcal{X},\Omega} ~~|| \mathcal{P}_{\Omega}(\mathcal{Y}- \mathcal{X})||_F^2 + \lambda \cdot \text{rank}(\mathcal{X}),~~~~
  s.t.
  ~~|\Omega| \leq M,
  \end{equation}
  where $\mathcal{X}$ is the decision variable, $\text{rank}(\cdot)$ refers to the tensor tubal-rank, $M$ is the sampling budget, and $\lambda$ is a regularization parameter. This approach aims to seek the simplest explanation fitting the samples.

{Clearly the performance depends on the type of sampling strategy used and the algorithm for solving the optimization problem. In this paper we will consider two kinds of sampling, namely, uniform random tubal-sampling and adaptive tubal-sampling.}

%% file: Sections/Baseline.tex
\section{Random tubal-sampling Approach}
\label{sect:baseline}

Before introducing the adaptive sampling approach, we first present a non-adaptive sampling approach, i.e., tensor completion via uniform tubal-sampling. We would like to point out again that tensor completion via uniform entry-wise sampling is well-studied \cite{Tensor2009SIAM,Tensor2011,Shuchin}, however, \emph{tubal-sampling} is required in RF fingerprint-based indoor localization. By uniform tubal-sampling, we mean that in (\ref{problem_formulation}) the subset $\Omega$ is simply chosen uniformly randomly (with replacement) from the grid map $\mathcal{G}$ with $|\Omega| \leq M$. {Similar to the matrix case, given a fixed $\Omega$  solving the optimization problem (\ref{problem_formulation}) is NP-hard. In this case one can relax the tubal-rank measure to a tensor nuclear norm (TNN) measure as proposed in \cite{Shuchin} and solve for the resulting convex optimization problem. The tensor nuclear norm as derived from the t-SVD is defined as follows.}

   \begin{definition}\label{lemma:tnn} \cite{Tensor2009SIAM,Shuchin}
   The tensor-nuclear-norm (TNN) denoted by $\| \mathcal{T} \|_{TNN}$ and defined as the sum of the singular values of all the frontal slices of $\breve{\mathcal{T}}$, i.e., $\| \mathcal{T} \|_{TNN} = \sum_{k=1}^{N_3} \|\breve{\mathcal{T}}(:,:,k)\|_{*}$, where $\|\cdot \|_{*}$ denotes the nuclear norm.
   \end{definition}

Extending the noiseless case considered in \cite{Shuchin} using the tensor-nuclear-norm  for tensor completion, in this paper we solve  for the noisy case (\ref{problem_formulation}) via,
   \begin{equation}\label{tensor_completion}
   \begin{split}
    \hat{\mathcal{T}} = \arg \min_{\mathcal{X}} &~~|| \mathcal{P}_{\Omega}(\mathcal{Y} - \mathcal{X}) ||_F^2 + \lambda || \mathcal{X} ||_{TNN}.\\
   \end{split}
   \end{equation}

  This optimization problem can be solved using ADMM, with modifications to the algorithm in \cite{Shuchin}.  Note that in \cite{Shuchin} the authors consider a random element-wise sampling whereas here we are considering a random tubal-sampling. It turns out that under tubal sampling the tensor completion by solving the convex optimization of Equation (\ref{tensor_completion}) splits into $N_3$ matrix completion problems in the Fourier domain. This observation leads us to derive the performance bounds as an extension of the results in matrix completion using the following notion of incoherency conditions.

  \begin{definition} (Tensor Incoherency Condition)
   Given the t-SVD of a tensor $\mc{T} = \mc{U} *\mc{S} *\mc{V}^\top$ with tubal-rank $r$, $\mc{T}$ is said to satisfy the tensor-incoherency conditions, if there exists $\mu_0 > 0$ such that for $k = 1,2,...,N_3$.
\begin{equation}\label{tensor_column_incoherency}
\begin{aligned}
&\text{(Tensor-column incoherence: )}\\~~&\mu(\mathcal{U}) \triangleq \frac{N_1}{r}\max_{i=1,...,N_1} \left\| \breve{\mc{U}}^{T}(:,:,k) \mathbf{e}_i \right\|_{F}^2 \leq \mu_0,\\
&\text{(Tensor-row incoherence: )}\\~~&\mu(\mathcal{V}) \triangleq \frac{N_2}{r} \max_{j=1,...,N_2} \left\| \breve{\mc{V}}^{T}(:,:,k) \mathbf{e}_j \right\|_{F}^{2} \leq \mu_0,
\end{aligned}
\end{equation}
where $\| \cdot\|_{F}$ denotes the usual Frobenius norm and $\mathbf{e}_i, \mathbf{e}_j$ are standard co-ordinate basis of respective lengths.
\end{definition}

We have the following result stated without proof.

{\begin{theorem}
Under random tubal sampling, if $\mu(\mathcal{U}), \mu(\mathcal{V}) \leq \mu_0$ then for  $M \geq C N \mu_0 r \log^2 N$ for some constant $C>0$ and $N = \max \{N_1,N_2\}$, solving the convex optimization problem of Equation (\ref{tensor_completion}) recovers the tensor with high probability (for sufficiently large $N$).
\end{theorem}}

{In contrast to this result we will see that adaptive sampling as we propose below requires almost the same sampling budget but \emph{only requires tensor column incoherence $\mu(\mathcal{U})$ to be small} (see Theorem \ref{recovery_error}),  which is one of the major gains of adaptive tubal-sampling over random tubal-sampling as well as over random element-wise sampling considered in \cite{Shuchin2015PAMI}. This gain is similar to the matrix case (under element-wise adaptive sampling) as shown in  \cite{Singh2013NIPS}. \emph{Adaptive sampling allows one to obtain more gains when the energy seems to be concentrated on few locations (for example, see Figure \ref{fig:scenario}), which the random sampling can miss!} Further, adaptive sampling reduces the number of measurements by a constant factor for the same accuracy. This is borne out by our experimental results. Although our approach is motivated by adaptive matrix sampling strategy in \cite{Singh2013NIPS}, we would like to point out that the performance bounds for the proposed adaptive strategy do not directly follow from the results in \cite{Singh2013NIPS} and requires a careful treatment.}


%% file: Sections/OurApproach.tex
\section{Proposed Adaptive Sampling Approach} \label{section:our_approach}

We begin by revisiting the problem and providing insights into the development of the proposed adaptive strategy.

    The problem (\ref{problem_formulation}) contains two goals:
    (1) For a given low-tubal-rank tensor $\mathcal{X}$, to select a set $\Omega$ with the smallest cardinality and the corresponding samples $\mathcal{Y}$, preserving most information of tensor $\mathcal{X}$, i.e., one can recover $\mathcal{X}$ from $\Omega$ and $\mathcal{Y}$.
    (2) For a given set $\Omega$ and samples $\mathcal{Y}$, to estimate a tensor $\mathcal{X}$ that has the least tubal-rank. However, these two goals are intertwined together and one cannot expect a computationally feasible algorithm to get the optimal solution. Therefore, we set $|\Omega| = M$ and seek to select a set $\Omega$ and the corresponding samples $\mathcal{Y}$ that span the low-dimensional tensor-column subspace of $\mathcal{T}$. The focus of this section is to design an efficient sampling scheme and to provide a bound on the sampling budget $M$.

    To achieve this, we design a two-pass sampling scheme inspired by \cite{VolSampling2006}. The proposed approach exploits adaptivity to identify entries that are highly informative for learning the low-dimensional {tensor}-subspace of the fingerprint data. The $1$st-pass sampling gathers general information about the region of interest, then the $2$nd-pass sampling concentrates on those more informative reference points. In particular the total sampling budget $M (< N_1N_2)$ is divided into $\delta M$ and $(1-\delta)M$ for these two sampling passes and $\delta$ is called {\em the allocation ratio}. In the 1st-pass sampling, we randomly sample $\delta M/N_2$ out of $N_1$ reference points in each column of $\mathcal{G}$. In the 2nd-pass sampling, the remaining $(1-\delta)M$ samples are allocated to those highly informative columns identified by the 1st-pass sampling. Finally, tensor completion on those $M$ RF fingerprints is performed to rebuild a fingerprint database.

{The provable optimality of this scheme rests on these three observations.
   \begin{itemize}
   \setlength\itemsep{-3pt}
   \item $\mathcal{T}$ is embedded in an $r$-dimensional tensor-column subspace $\mathcal{U}$, $r \ll \min(N_1,N_2)$;
   \item Learning $\mathcal{U}$ requires to know only $r$ \emph{linearly independent} lateral slices;\footnote{{Note: A collection of $r$ lateral slices $\mathcal{U}(:,j,:), j = 1,...,r$ are said to be linearly independent (in the proposed setting) if $ \sum_{j=1}^{r}\mathcal{U}(:,j,:)\ast \mathbf{c}_j = 0 \implies \mathbf{c}_j = 0, \forall j$.} }
   \item Knowing $\mathcal{U}$, randomly sampling a few tubes of the $j$-th column is enough to reliably estimate the lateral slice $\mathcal{T}(:,j,:)$;
   \end{itemize}}
  However, we do not know the value of $r$ {\em a priori} nor the independency between any two lateral slices. {Given a current estimate $\mathcal{U}$ of the tensor-column subspace, following \cite{VolSampling2006}} one can \emph{adaptively} sample each column according to the probability distribution where the probability $p_j$ of sampling the $j$-th lateral slice is proportional to $||\mathcal{P}_{\mathcal{U}^{\bot}} \mathcal{T}(:,j,:)||_F^2$, i.e, $p_j = \frac{||\mathcal{P}_{\mathcal{U}^{\bot}} \mathcal{T}(:,j,:)||_F^2}{||\mathcal{P}_{\mathcal{U}^{\bot}} \mathcal{T}||_F^2}$. Updating the estimate of $\mathcal{U}$ iteratively, when $cr$ ($c>1$ is a small constant) columns are sampled, we can expect that with high probability, $||\mathcal{P}_{\mathcal{U}^{\bot}} \mathcal{T}(:,j,:)||_F^2 =0$, $\forall j$. Note that $\mathcal{P}_{\mathcal{U}^{\bot}}(\cdot)$ denotes projection onto the orthogonal space of $\mathcal{U}$; in t-product form, $\mathcal{P}_{\mathcal{U}} = \mathcal{U} \ast (\mathcal{U}^T \ast \mathcal{U})^{-1} \ast \mathcal{U}^T$, $\mathcal{P}_{\mathcal{U}^{\bot}} = \mathcal{I} - \mathcal{P}_{\mathcal{U}}$, and $\mathcal{U}^T \ast \mathcal{U} $ is invertible and can be computed according to Definition \ref{inverse} and Remark \ref{computation}.

The challenge is that we cannot have the exact sampling probability $p_j$ without sampling all reference points of the grid map $\mathcal{G}$. Nevertheless, exploiting the spatial correlation, one can estimate the sampling probability from missing data (sub-sampled data), as $\hat{p}_j = \frac{||\mathcal{P}_{\mathcal{U}^{\bot}}(\mathcal{T}(:,\Omega_j^1,:))||_F^2}{|| \mathcal{P}_{\mathcal{U}^{\bot}}(\mathcal{T}(:,\Omega^1,:))||_F^2}$ in Algorithm \ref{alg:puncturing}. Essentially { Under the incoherency conditions of Equation (\ref{tensor_column_incoherency}) and for sufficiently large $M$ we show in Lemma \ref{lemma:sampling_budget} that $\hat{p}_j$ is a good estimation of $p_j$. Then using these estimates one can show that adaptive sampling scheme of the second pass of Algorithm 1 succeeds with high probability in estimating the correct tensor column-subspace. This is the content of Lemma \ref{lemma_singleround} which analyzes a single step of second pass sampling using which the main Theorem \ref{recovery_error} follows.}

We now outline the details of the algorithm in the next section.

\subsection{Tensor Completion with Adaptive Tubal-Sampling}

   \begin{figure}[t]
   \centering
   \includegraphics[width=0.25\textwidth]{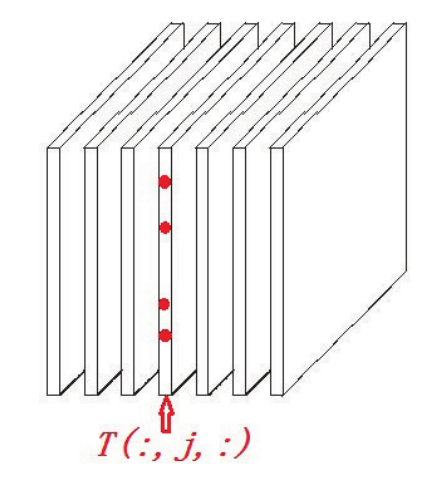}
   \caption{After adaptively building an estimate of $\mathcal{U}$ using two pass sampling approach of Algorithm \ref{alg:puncturing}, in the last step of the algorithm, the engineer performs site survey at a few random selected reference points (with red color) of the $j$-th column of $\mathcal{G}$. Knowing the $r$-dimensional tensor-subspace $\mathcal{U}$, one can recover the $j$-th lateral slice $\mathcal{T}(:,j,:)$.} \label{fig:tensor_slice}
   \end{figure}

    The pseudo-code of our adaptive tubal-sampling approach is shown in Algorithm \ref{alg:puncturing}. The inputs include the grid map $\mathcal{G}$, the sampling budget $M$, the size of the tensor, $N_1,N_2,N_3$, the allocation ratio $\delta$, and the number of iterations $L$. The algorithm consists of three steps. The $1$st-pass sampling is a uniform tubal-sampling, 
    while the $2$nd-pass sampling outputs an estimate $\hat{\mathcal{U}}$ of the tensor-column subspace $\mathcal{U}$ in $L$ rounds, as explained below.

\begin{algorithm}[h]
  \caption{Tensor completion based on adaptive sampling}
  \label{alg:puncturing}
  \begin{algorithmic}
  \STATE \textbf{Input}: parameters $\mathcal{G},M, N_1,N_2,N_3,\delta,L$
  \STATE \textbf{1st-pass sampling}:
  \STATE Uniformly sample $\delta M / N_2$ reference points from each column of $\mathcal{G}$, denoted as $\Omega_j^1$. $\Omega^1 = \bigcup\limits_{j=1}^{N_2} \Omega_j^1$.
  \STATE \textbf{2nd-pass sampling}:
  \STATE $\hat{\mathcal{U}} \leftarrow \emptyset$, $\Pi \leftarrow \{1,2,...,N_2\}$.
  \FOR{$l = 1 : L$} 
       \STATE Estimate $\hat{p}_j = \frac{||\mathcal{P}_{\hat{\mathcal{U}}^{\bot}}(\mathcal{T}(:,\Omega_j^1,:))||_F^2}{|| \mathcal{P}_{\hat{\mathcal{U}}^{\bot}}(\mathcal{T}(:,\Omega^1,:))||_F^2}$, $\forall j \in \Pi$.
       \STATE Sample $s= \frac{(1-\delta)M }{N_2L}$ columns of $\mathcal{G}$ according to $\hat{p}_j$ for $\forall j \in \Pi$, denoted as $\Pi_s^l$.
       \STATE Calculate $\mathbb{U} \leftarrow \frac{\mathcal{P}_{\hat{\mathcal{U}}^{\bot}}(\mathcal{T}(:,\Pi_s^l,:))}{|| \mathcal{P}_{\hat{\mathcal{U}}^{\bot}}(\mathcal{T}(:,\Pi_s^l,:))||_F}$, and update $\hat{\mathcal{U}} \leftarrow \hat{\mathcal{U}} \cup \mathbb{U}$ (concatenate $\mathbb{U}$ to $\hat{\mathcal{U}}$).
       \STATE $\Pi \leftarrow \Pi / \Pi_s^l$ (set subtraction).
  \ENDFOR
¡¡\STATE Approximate each lateral slice $\mathcal{T}(:,j,:)$ with
    $\hat{\mathcal{T}}(:,j,:) = \hat{\mathcal{U}} \ast (\hat{\mathcal{U}}_{\Omega_j}^T \ast \hat{\mathcal{U}}_{\Omega_j})^{-1} \ast \hat{\mathcal{U}}_{\Omega_j}^T \ast \mathcal{T}(\Omega_j,j,:)$.
  \end{algorithmic}
  \vspace{-2pt}
\end{algorithm}

\subsubsection{$1$st-Pass Sampling}

   First, we gather general information of the whole region of interest, applying a uniform random sampling to avoid spatial bias. Denote these sampled $\delta M$ reference points as $\Omega^{1}$, the sampled reference points in the $j$-th column as $\Omega_j^{1}$ with $m=|\Omega_j^{1}|=\delta M/N_2$, and the corresponding fingerprints as $\mathcal{T}(\Omega_j^{1},j,:)$.

\subsubsection{$2$nd-Pass Sampling}
   Initialize with $\hat{\mathcal{U}} \leftarrow \emptyset$, $\Pi \leftarrow \emptyset$. In each round, we estimate the sampling probability $\hat{p}_j = \frac{||\mathcal{P}_{\tilde{\mathcal{U}}^{\bot}}(\mathcal{T}(:,\Omega_j^1,:))||_F^2}{|| \mathcal{P}_{\hat{\mathcal{U}}^{\bot}}(\mathcal{T}(:,\Omega^1,:))||_F^2}$, and choose $s$ columns of $\mathcal{G}$ according to the probability $\hat{p}_j$. Then, we calculate an intermediate subspace $\mathbb{U}$ that is the space spanned by these $s$ columns and lies outside of $\hat{\mathcal{U}}$, and then update the subspace $\hat{\mathcal{U}}$. 
   $\Pi_s^l$ denotes these $s$ lateral slices, $\mathcal{T}(:,\Pi_s^l,:)$ denotes these RF fingerprints in the $l$-th round, corresponding to these $s$ columns of $\mathcal{G}$.


\subsubsection{Estimation}
   Let $\Omega$ denote all sampled reference points (including $\Omega^1$), $\Omega_j$ denotes the sampled reference points in the $j$-th column of $\mathcal{G}$, and $\mathcal{U}_{\Omega_j}$ denotes the tensor organized by the horizontal slices of $\mathcal{U}$ indicated by $\Omega_j$. Define the projection operator $\mathcal{P}_{\mathcal{U}_{\Omega_j}} = \mathcal{U}_{\Omega_j} \ast (\mathcal{U}_{\Omega_j}^T \ast \mathcal{U}_{\Omega_j})^{-1} \ast \mathcal{U}_{\Omega_j}^T$. After $L$ rounds, we can obtain an fairly accurate estimation of $\mathcal{U}$. To understand the estimator, considering the noiseless case. Since $\mathcal{T}$ lies in $\mathcal{U}$, we have $||\mathcal{T}(\Omega_j,j,:) - \mathcal{P}_{\mathcal{U}_{\Omega_j}} \ast \mathcal{T}(\Omega_j,j,:)  ||=0$, i.e., $\mathcal{T}(\Omega_j,j:) = \mathcal{U}_{\Omega_j} \ast (\mathcal{U}_{\Omega_j}^T \ast \mathcal{U}_{\Omega_j})^{-1} \ast \mathcal{U}_{\Omega_j}^T \ast \mathcal{T}(\Omega_j,j,:)$. And according to Definition \ref{tsvd}, we know that $\mathcal{T}(\Omega_j,j,:) = \mathcal{U}_{\Omega_j} \ast \Theta \ast \mathcal{V}^T(:,j,:)$. Therefore, using the estimate $\hat{\mathcal{U}}$, we approximate each lateral slice $\mathcal{T}(:,j,:)$ with $\hat{\mathcal{T}}(:,j,:) = \hat{\mathcal{U}} \ast (\hat{\mathcal{U}}_{\Omega_j}^T \ast \hat{\mathcal{U}}_{\Omega_j})^{-1} \ast \hat{\mathcal{U}}_{\Omega_j}^T \ast \mathcal{T}(\Omega_j,j,:)$ and concatenate these estimates to form $\hat{\mathcal{T}}$.

\section{Performance Bounds}

   For performance guarantee, we are interested in the recovery error and required sampling budget. We prove that Algorithm \ref{alg:puncturing} has bounded recovery error and achieves near-optimal sampling budget. Since we use the estimated sampling probability in the $2$nd-pass sampling, we also prove that our estimates $\hat{p}_j$ are relatively close to $p_j$.

{\textbf{STEP 1}} - First, we analyze a single round of the $2$nd-pass sampling. {Lemma \ref{lemma_singleround} states that if the probability estimates are to within a constant tolerance of the true estimates (see Equation \ref{sampling_probability}), then sampling $s$ columns of $\mathcal{G}$ according to $\hat{p}_j$ (estimated based on samples obtained in the $1$st-pass sampling) will minimize the residual error within $\hat{\mathcal{U}}$ at rate $\frac{1}{s}$ ($s \geq 1$).}. The second term in the right-hand side of (\ref{error_singleround}) denotes the residual error outside of $\hat{\mathcal{U}}$, which remains unreduced. Note that without any {\em prior} information (i.e., $\hat{\mathcal{U}}$ and $\hat{p}_j$), sampling additional $s$ columns of $\mathcal{G}$ will reduce the residual error at rate $\frac{s}{N_2}$ which is $ \ll \frac{1}{s}$ as $s \ll \sqrt{N_2}$ in Algorithm \ref{alg:puncturing}. Therefore, Lemma \ref{lemma_singleround} is the key to efficiently reduce recovery error. Note that  $\mathcal{E} = \mathcal{P}_{\hat{\mathcal{U}}^{\bot}}(\mathcal{T})$ denotes the projection onto the orthogonal complement of $\hat{\mathcal{U}}$. See Appendix A for its proof.

  \begin{lemma}\label{lemma_singleround}
   Let $\mathcal{T} = \mathcal{U} \ast \Theta \ast \mathcal{V}^T \in \mathbb{R}^{N_1 \times N_2 \times N_3}$ with tensor tubal-rank $r$, and $\hat{\mathcal{U}} \in \mathbb{R}^{N_1 \times r \times N_3}$ represent the estimated tensor-column subspace in a round. Let $\mathcal{E} = \mathcal{P}_{\hat{\mathcal{U}}^{\bot}}(\mathcal{T}) = \mathcal{T} - P_{\hat{\mathcal{U}}}(\mathcal{T})$, and $\mathcal{S}$ be $s$ randomly selected lateral slices of $\mathcal{T}$ (as indicated by $\Pi_s^l$), sampled according to the distribution $\hat{p}_j = \frac{||\mathcal{P}_{\hat{\mathcal{U}}^{\bot}}(\mathcal{T}(:,\Omega_j^1,:))||_F^2}{|| \mathcal{P}_{\hat{\mathcal{U}}^{\bot}}(\mathcal{T}(:,\Omega^1,:))||_F^2}$. If there exist constants $\alpha_1,\alpha_2 \in \mathbb{R}$, such that:
   \begin{equation}\label{sampling_probability}
   \frac{1 - \alpha_1}{1 + \alpha_2} \cdot \frac{||\mathcal{E}(:,j,:)||_F^2}{||\mathcal{E}||_F^2} \leq \hat{p}_j \leq
   \frac{1 + \alpha_2}{1 - \alpha_1} \cdot \frac{||\mathcal{E}(:,j,:)||_F^2}{||\mathcal{E}||_F^2},
   \end{equation}
   then, with probability $\geq 1 - \rho$ we have:
   \begin{equation}\label{error_singleround}
   ||\mathcal{T} - P_{\hat{\mathcal{U}} \cup \text{t-span}(\mathcal{S}), r}(\mathcal{T})||_F^2 \leq \frac{r}{s\rho} \cdot \xi_0 \frac{1 + \alpha_2}{1 - \alpha_1} \cdot ||\mathcal{E}||_F^2  + ||\mathcal{T} - \mathcal{T}_r||_F^2,
   \end{equation}
   where $\xi_0$ is a constant such that  \begin{align}
    \sum_{j=1}^{N_2} |\breve{\mc{V}}^{T}(i,j,k)|^2 \leq \xi_0,
   \end{align} $\text{t-span}(\mathcal{S})$ denotes the space spanned by the slices of $\mathcal{S}$, and $P_{\mathcal{H},r}(\cdot)$ denotes a projection on to the best $r$-dimensional subspace of $\mathcal{H}$. 
   \end{lemma}

{\textbf{STEP 2}} - {Now note that we perform the adaptive sampling scheme for $L$ rounds. For this} Lemma \ref{error_allround} below gives an induction argument to chain {Lemma \ref{lemma_singleround}} across all rounds of the $2$nd-pass sampling.


Setting $\epsilon < 1$ and $L$ sufficiently large, Lemma \ref{error_allround} states that our two-pass sampling scheme can approximate the low tubal-rank tensor $\mathcal{T}$ with error comparable to that of the best rank-$r$ approximation $\mathcal{T}_{r}$. This indicates that the $2$nd-pass sampling estimates $\mathcal{U}$ with high accuracy.


   \begin{lemma}\label{error_allround}
   Suppose that (\ref{sampling_probability}) holds with $\frac{1 + \alpha_2}{1 - \alpha_1} \leq c$ for some constant $c$ in each round. Let $\mathcal{S}_1,\mathcal{S}_2,...,\mathcal{S}_L$ denote the sets of slices selected at each round and set $s = \frac{Lcr \xi_0}{\rho \epsilon}$. Then with probability $\geq 1-\rho$ we have:
   \begin{eqnarray}
   &&||\mathcal{T} - \mathcal{P}_{\hat{\mathcal{U}}}(\mathcal{T}) ||_F^2 = ||\mathcal{T} - P_{\cup_{i=1}^{L} \mathcal{S}_i, r}(\mathcal{A})||_F^2 \nonumber\\
   &&\leq \frac{1}{1-\epsilon} ||\mathcal{T} - \mathcal{T}_r||_F^2 + \epsilon^L ||\mathcal{T}||_F^2.
   \end{eqnarray}
   \end{lemma}
{\begin{proof} The proof follows along the same lines as the proof for the adaptive matrix completion case in \cite{Singh2013NIPS}, by applying the matrix result to $N_3$ frontal slices of the tensor in the Fourier domain $\breve{\mathcal{T}}$. \end{proof}}

{Now note that Lemma \ref{lemma_singleround} requires, that the estimates of $\hat{p}_j$ are close to the actual values by a factor at most $c = \frac{1 + \alpha_2}{1 - \alpha_1}$. For this Lemma \ref{lemma:sampling_budget} states that this condition holds with high probability for $\alpha_1 = -1, \alpha_2 = 4$, and $c=\frac{5}{2}$}.

   \begin{lemma}\label{lemma:sampling_budget}
   Let $\mathcal{E} = \mathcal{P}_{\hat{\mathcal{U}}^{\bot}}(\mathcal{T})$ and $\mathcal{T}$ satisfies the tensor-column incoherence condition (\ref{tensor_column_incoherency}) with $\mu(\mathcal{U}) \leq \mu_0$, with probability $\geq 1 - 6 \rho$ we have
   \begin{equation}\label{probability_estimation}
   \frac{2}{5} \cdot \frac{||\mathcal{E}(:,j,:)||_F^2}{||\mathcal{E}||_F^2} \leq \hat{p}_j \leq
   \frac{5}{2} \cdot \frac{||\mathcal{E}(:,j,:)||_F^2}{||\mathcal{E}||_F^2}.
   \end{equation}
   as long as the expected sampling budget $M$ satisfies:
   \begin{eqnarray}\label{sampling_number}
   M &\geq& CrN_2 (\mu_0\log(N_2/\rho)\log(rN_2/\rho) \nonumber\\&& + (\lceil \log(N_1 N_2N_3)\rceil)^2 /(\rho \epsilon)), 
   \end{eqnarray}
 for some  $C >0$, and thus  is of order $O(Nr\log^2 N)$ with $N = \max(N_1,N_2)$.
   \end{lemma}
   \begin{proof}  This result relies on three conditions: the incoherence of each lateral slice $\mathcal{T}(:,j,:)$, the tensor-column incoherence of $\mathcal{U}$ - Equation (\ref{tensor_column_incoherency}), and Lemma \ref{space_detection}; those three conditions' failure probabilities are less than $\rho$, $\rho$, and $4\rho$, respectively; therefore, Lemma \ref{lemma:sampling_budget} holds with probability $\geq 1- 6\rho$. See Appendix C for a complete proof. \end{proof}

   \begin{remark}
   For $\mathcal{T} \in \mathbb{R}^{N_1 \times N_2 \times N_3}$ with tensor tubal-rank $r$, its t-SVD in Definition \ref{tsvd} indicates that the degree of freedom (in terms of non-zero vectors) is $N_1 r + N_2 r + r$ which is of order $O(Nr)$ with $N = \max(N_1,N_2)$. Therefore, our sampling budget is near-optimal within a factor of $\log^2 N$.
   \end{remark}

{\textbf{STEP 3}} - {Finally, we analyze the estimation process (Last step of the Algorithm) in Lemma \ref{lemma:each_slice}}. It states that our estimator outputs each lateral slice with bounded error, which is comparable to the energy outside of $\hat{\mathcal{U}}$, i.e., $||\mathcal{P}_{\hat{\mathcal{U}}^{\bot}} \mathcal{T}(:,j,:)||_F$.

   \begin{lemma}\label{lemma:each_slice}
   For $\mathcal{T} \in \mathbb{R}^{N_1 \times N_2 \times N_3}$ with tensor tubal-rank $r$, let $\mathcal{T}(:,j,:)$ denote the $j$-th lateral slice and $\hat{\mathcal{U}}$ denote the tensor-column subspace at the $L$-th round of the $2$nd-pass sampling. Algorithm \ref{alg:puncturing} estimates $\mathcal{T}(:,j,:)$ as $\hat{\mathcal{T}}(:,j,:) = \hat{\mathcal{U}} \ast (\hat{\mathcal{U}}_{\Omega_j}^T \ast \hat{\mathcal{U}}_{\Omega_j})^{-1} \ast \hat{\mathcal{U}}_{\Omega_j}^T \ast \mathcal{T}(\Omega_j,j,:)$. Then with probability $\geq 1 - 2\rho$,
   \begin{equation}
   ||\mathcal{T}(:,j,:) - \hat{\mathcal{T}}(:,j,:) ||_F^2 \leq \left(1 + \frac{r\mu(\hat{\mathcal{U}}) \beta}{m(1-\gamma)^2} \right) ||\mathcal{P}_{\hat{\mathcal{U}}^{\bot}} \mathcal{T}(:,j,:)||_F^2.
   \end{equation}
   where $\beta = (1 + 2 \sqrt{\log(1/\rho)})^2$, $\gamma = \sqrt{\frac{8r\mu(\hat{\mathcal{U}})}{3m}\log(2r/\rho)}$, and $\mu(\hat{\mathcal{U}})$ is defined in (\ref{tensor_column_incoherency}).
   \end{lemma}

   With our choice of $m$ (and correspondingly $M$ related via $\delta M = m N_2$), $(1 + \frac{r\mu(\hat{\mathcal{U}}) \beta}{m(1-\gamma)^2})$ is smaller than $5/4$ as shown in Appendix \ref{proof:sampling_budget}. Therefore, $|| \hat{\mathcal{T}}- \mathcal{T}'||_F^2 = \sum\limits_{j=1}^{N_2} ||\mathcal{T}(:,j,:) - \hat{\mathcal{T}}(:,j,:) ||_F^2 \leq \frac{5}{4} ||\mathcal{T} - \mathcal{P}_{\hat{\mathcal{U}}}(\mathcal{T})||_F^2$, combining Lemma \ref{error_allround} and Lemma \ref{lemma:sampling_budget}, (see also \cite{Singh2013NIPS}), we have Lemma \ref{lemma_ap}. Note that Lemma \ref{error_allround}, Lemma \ref{lemma:sampling_budget}, and Lemma \ref{lemma:each_slice} each has failure probability less than $\rho$, $2\rho$ and $6\rho$, therefore, Lemma \ref{lemma_ap} has success probability $\geq 1- 9 \rho$.

   \begin{lemma}\label{lemma_ap}
   Assume that $\mathcal{T}' = \mathcal{T} + \mathcal{N}$, $\mathcal{T}$ has tensor tubal-rank $r$, tensor-column incoherence $\mu(\mathcal{U}) \leq \mu_0$, $\mathcal{N}(i,j,k) \sim N(0,\sigma^2)$, and for all $i$, $\sum_{j=1}^{N_2} |\breve{\mc{V}}^{T}(i,j,k)|^2 \leq \xi_0$. Then for $\rho, \epsilon \in (0,1)$, sample $s = \frac{5Lr \xi_0}{2\rho \epsilon}$ columns of $\mathcal{G}$ each round, after $L$ 
   rounds, compute $\hat{\mathcal{T}}$ as described. Then with probability $\geq 1 - 9 \rho$,
   \begin{equation}
   || \hat{\mathcal{T}}- \mathcal{T}'||_F^2 \leq 5/4 \left( \frac{1}{1 - \epsilon} ||\mathcal{T}' - \mathcal{T}_r'||_F^2 + \epsilon^L||\mathcal{T}'||_F^2 \right).
   \end{equation}
   \end{lemma}

   While Lemma \ref{lemma_ap} bounds the difference between $\hat{\mathcal{T}}$ and $\mathcal{T}'$, we then bound the difference between $\hat{\mathcal{T}}$ and $\mathcal{T}$, which measures the recovery error of Algorithm \ref{alg:puncturing}. This error consists of two parts: the first term measures the performance of our estimation process, and the second term is essentially $||\mathcal{N}_{\Omega}||_F^2$ which measures the effect of noise in the samples $\mathcal{Y}$. {Combining these results we have the following main Theorem}.

\begin{theorem}\label{recovery_error}
   Under the partial observation model $\mathcal{P}_{\Omega}(\mathcal{T} + \mathcal{N})$, where $\mathcal{T} \in \mathbb{R}^{N_1 \times N_2 \times N_3}$, $\mathcal{N}(i,j,k) \sim N(0,\sigma^2)$. Assume that $\mathcal{T}$ has tubal-rank $r$, $\sum_{j=1}^{N_2} |\breve{\mc{V}}^{T}(i,j,k)|^2 \leq \xi_0, \forall i$, and  tensor-column incoherence $\mu(\mathcal{U}) \leq \mu_0$. Then for $\rho, \epsilon \in (0,1)$, $L = \lceil \log_{1/\epsilon}(N_1 N_2N_3) \rceil$ and 
   \begin{eqnarray}M &=& \max \left\{\frac{5N_2 L^2 r \xi_0}{2(1-\delta) \rho \epsilon},\right.\nonumber\\
    && \left.CrN_2 (\mu_0\log(N_2/\rho)\log(rN_2/\rho) + L^2 /(\rho \epsilon)) \right\},
    \end{eqnarray}with probability $\geq 1 - 9\rho$, there exist two constants $c_1,c_2$ such that the estimation error of Algorithm 1 obeys,
   \begin{equation}
   || \mathcal{T} - \hat{\mathcal{T}}||_F^2 \leq \frac{c_1}{N_1N_2N_3} ||\mathcal{T}||_F^2 + c_2(MN_3 + \sqrt{8MN_3})\sigma^2. 
   \end{equation}
   \end{theorem}

As a consequence of this theorem, for example, assuming that the $\ell_2$-norm of each fingerprint is approximately the same, say $C$, then our algorithm guarantees that the recovery error of each fingerprint in $\ell_2$-norm will be bounded by $\sqrt{\frac{c_1 C^2}{N_1N_2N_3} + \frac{c_2(MN_3 + \sqrt{8MN_3})\sigma^2}{N_1N_2}}$ and the relative error is bounded by $\sqrt{\frac{c_1}{N_1N_2N_3} + \frac{c_2(MN_3 + \sqrt{8MN_3})\sigma^2}{N_1N_2C^2}}$. Since $N_1,N_2$ are relatively large, $N_3 \ll \min(N_1,N_2)$ and $M$ is provided in Lemma \ref{lemma:sampling_budget} to be in order of $O(N r\log^2 N)$ with $N = \max(N_1,N_2)$, therefore, the relative error is small.

%% file: Sections/Simulation.tex
\section{{Performance Evaluation}}

  We are interested in two kinds of performance: recovery error and localization error. Varying the sampling rate as $10\% \sim 90\%$, we quantify the recovery error in terms of {\em normalized square of error (NSE)} for entries that are not sampled, i.e., recovery error for set $\Omega^{c}$. The NSE is defined as:
  \begin{equation}
  NSE = \frac{\sum_{(i,j)\in \Omega^{c}} ||\hat{\mathcal{T}}(i,j,:) - \mathcal{T}(i,j,:) ||_F^2 }{ \sum_{(i,j)\in \Omega^{c}} ||\mathcal{T}(i,j,:)||_F^2 },
  \end{equation}
  where $\hat{\mathcal{T}}$ is the estimated tensor, $\Omega^{c}$ is the complement of set $\Omega$.

  In the simulations, we uniformly select $500$ testing points within the selected region and then using the classic localization schemes to perform localization estimation. We measure the localization error as the Euclidean distance between the estimated location and the actual location of the testing point, i.e.,
  $d_e = \sqrt{(\hat{x}_1 - x_1)^2 + (\hat{x}_2 - x_2)^2}$.

{For tensor recovery, we consider three algorithms, tensor completion (TC) \emph{under uniformly random tubal-sampling} and using the algorithm proposed in \cite{Shuchin,Shuchin2015PAMI}, using the \emph{face-wise} matrix completion (MC) algorithm in \cite{Kong2013INFOCOM}, and tensor completion via matricization or flattening (MC-flat) \cite{Yamada} under \emph{uniform element-wise sampling} of the 3D tensor, using the AltMin algorithm for matrix completion \cite{AltMin_ACM}. We subsequently use the completed RSS map for localization and compare the error in location estimates.}

\subsection{Experiment Setup - Model-based Data}

  \begin{figure}[t]
  \centering
  \includegraphics[trim=0in .2in 1.6in 0in, clip,width=0.38\textwidth]{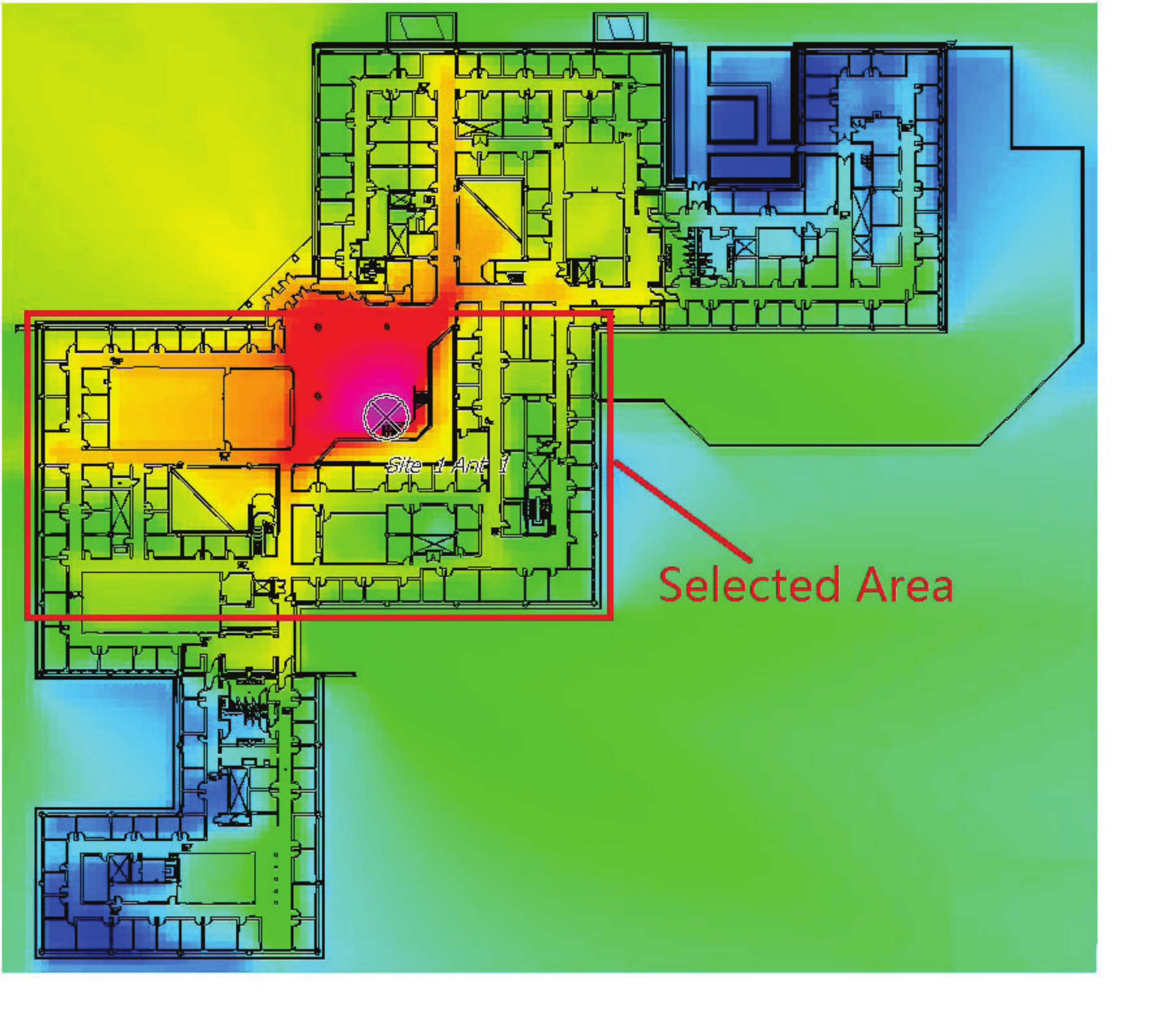}
  \includegraphics[trim=1in .6in 1in 0in, clip,width=0.1\textwidth]{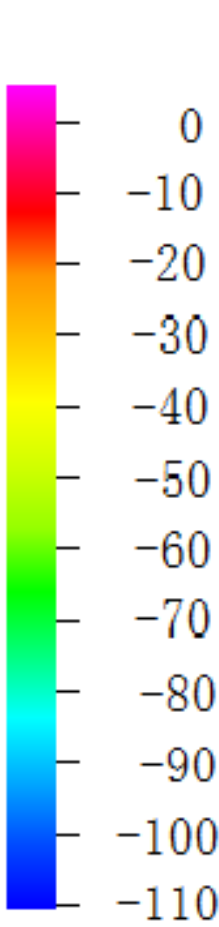}
  \caption{We select a rectangle area within an office building for simulations. RSS is measured in dBm.}\label{fig:scenario}
  \end{figure}


  We select a region of $47.5$m $\times$ $59.7$m in a real office building, as shown in Fig. \ref{fig:scenario}. It is divided into a $476 \times 598$ grid map. There are $15$ access points randomly deployed within this region.
  The indoor radio channel is characterized by multi-path propagation with dominant propagation phenomena: the shadowing of walls, wave guiding effects in corridors due to multiple reflections, and diffractions around vertical wedges.
  {The ray-tracing based model \cite{SRT,Ray_tracing_TWC2009} is adopted, which considers all these effects leading to highly accurate prediction results. Further, the model parameters were found by supervised learning with the real collected measurements using a professional software \cite{SRT}.}
  We generated a $476 \times 598 \times 15$ RSS tensor as the {\em ground truth} for our simulations. Note that the RSS values are measured in dBm. For example, the RSS radio map for the $5$-th and $15$-th access points are shown in Fig. \ref{fig:rss-map}.

  \begin{figure}[t]
  \centering
  \includegraphics[width=0.48\textwidth]{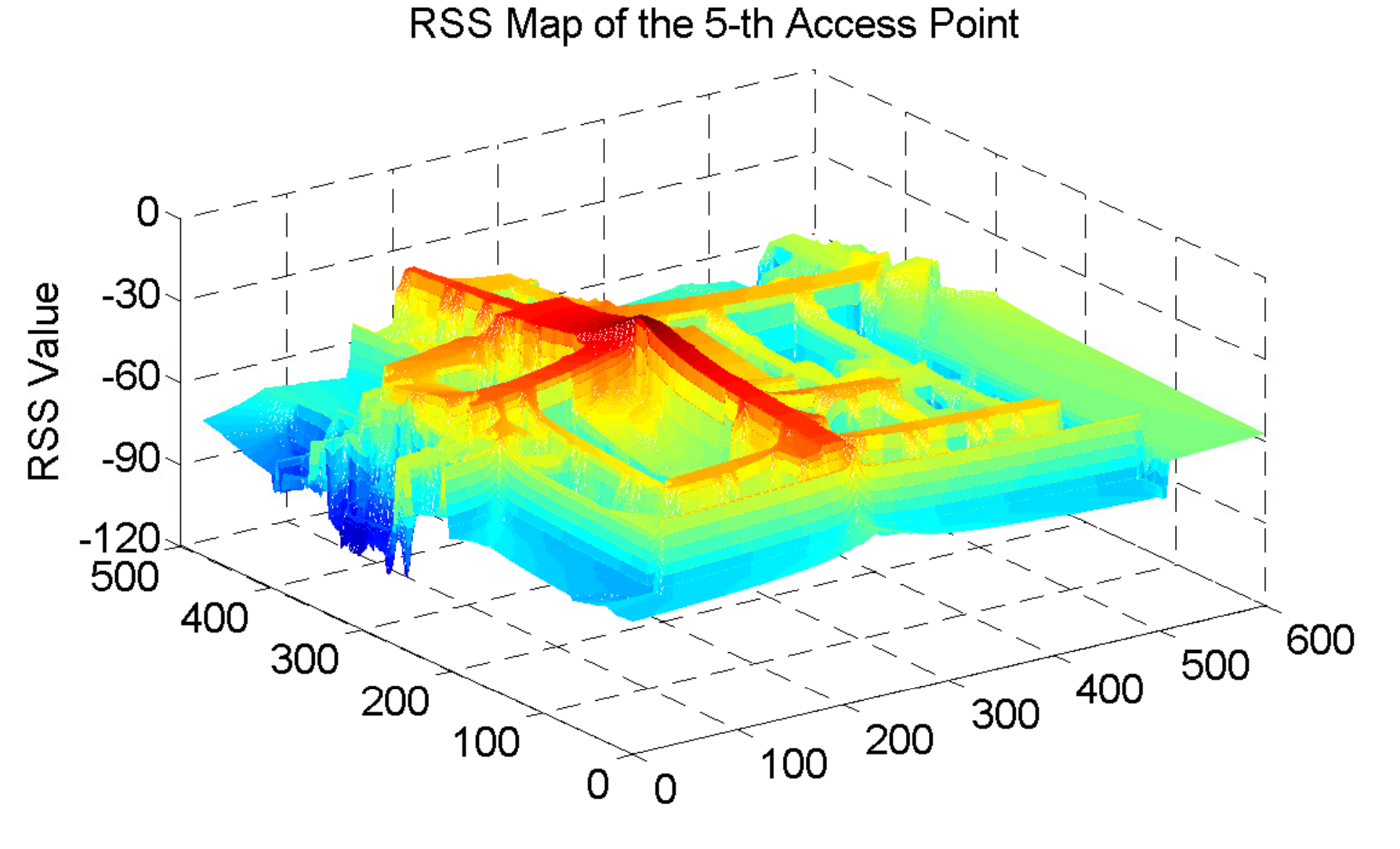}
  \includegraphics[width=0.48\textwidth]{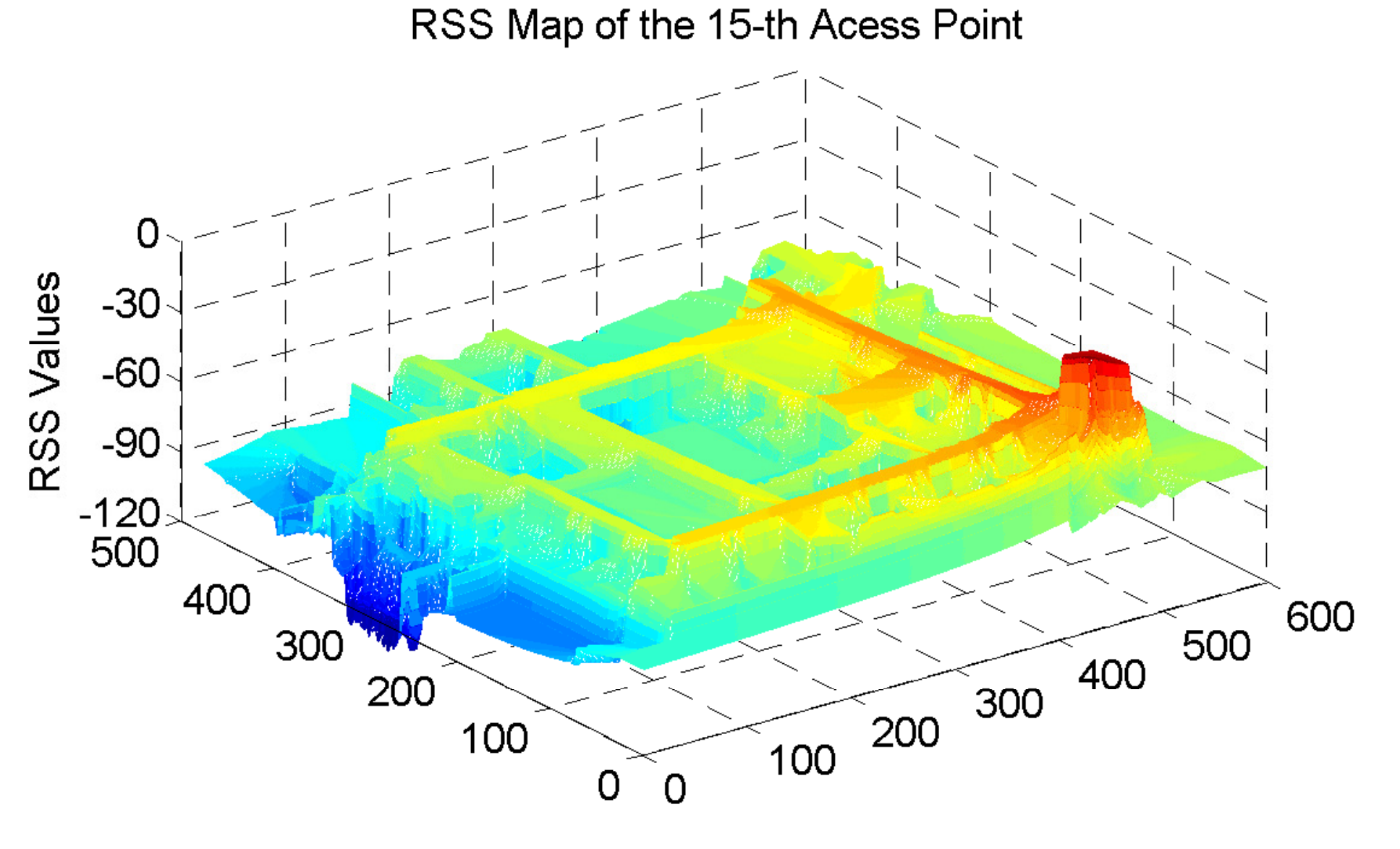}
  \caption{The RSS radio map of the $5$-th and $15$-th access points.}\label{fig:rss-map}
  \end{figure}


\noindent {\em Radio Map Recovery Performance: }  Fig. \ref{fig:recovery} shows the RSS tensor recovery performance for varying sampling rate. Compared schemes are matrix completion and tensor completion via uniform sampling, and adaptive sampling with allocation ratio $\delta = 1/4$ and $\delta = 1/2$. We find that all tensor approaches are better than matrix completion, this is because tensor exploits the cross correlations among access points while matrix completion only takes advantage of correlation within each access point. Both AS schemes outperform tensor completion via uniform sampling since adaptivity can guide the sampling process to concentrate on more informative entries. Allocating equal sampling budget for the $1$st-pass and the $2$nd-pass gives better performance than uneven allocation. This shows that the $1$st-pass and the $2$nd-pass have equal importance. The proposed scheme (AS with $\delta = 1/2$) rebuilds a fingerprint data with $5\%$ error using less than $30\%$ samples. 

  \begin{figure}[t]
  \centering
  \includegraphics[width=0.5\textwidth]{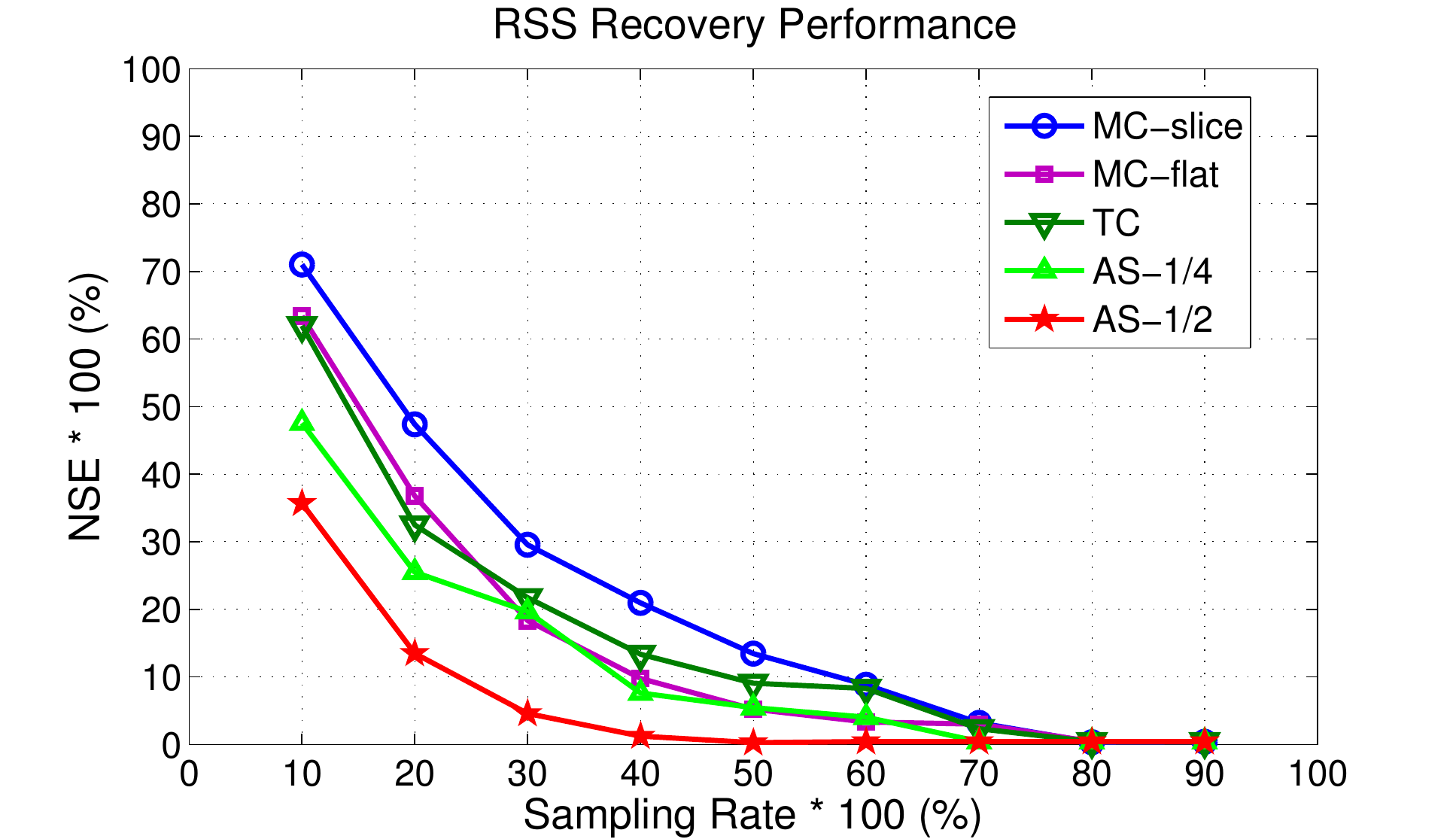}
  \caption{Tensor recovery for varying sampling rate. MC denotes matrix completion via uniform sampling, TC denotes tensor completion via uniform sampling, while AS-$1/4$ and AS-$1/2$ denote our adaptive sampling scheme with allocation ratios $\delta = 1/4$ and $\delta = 1/2$, respectively. {MC-flat denotes completion by \emph{flattening} the 3-D data into a matrix of locations $\times$ APs and using matrix completion, an approach followed in \cite{Nikitaki1}.}}\label{fig:recovery}
  \end{figure}

{\subsection{Experiment Setup - Real data}
We collected a WiFi RSS data set in the same office. The data set contains $89$ selected locations and $31$ access points. Since the locations are not exactly on a grid, we set the grid size to be $3$m $\times$ $5$m, and apply the KNN method to extract a full third-order tensor as the ground truth. To be specific, for each grid point, we set its RSS vector by averaging the RSS vectors from the nearest three ($k=3$) locations. The ground truth tensor has dimension $10 \times 10 \times 31$. This tensor serves as a complement to our model-generated data while in the next section, we want to test the localization performance at a finer granularity and covering the whole region of interest.}

\noindent {\em Radio Map Recovery Performance: } Fig. \ref{fig:real_data_recovery} shows the RSS tensor recovery performance for the real-world data set. First, compared with Fig. \ref{fig:recovery}, we see that the recovery performance on real-world data is consistent with that of simulated data.
  Second, for real-world data set, tensor model is superior to matrix model. In our case, a major ingredient for the recovery improvement may be the large number of access points (i.e., $31$), compared with the dimension of the grid (i.e., $10 x 10$). Third, as expected, the propose adaptive scheme achieves better recovery performance.

  \begin{figure}[t]
  \centering
  \includegraphics[width=0.5\textwidth]{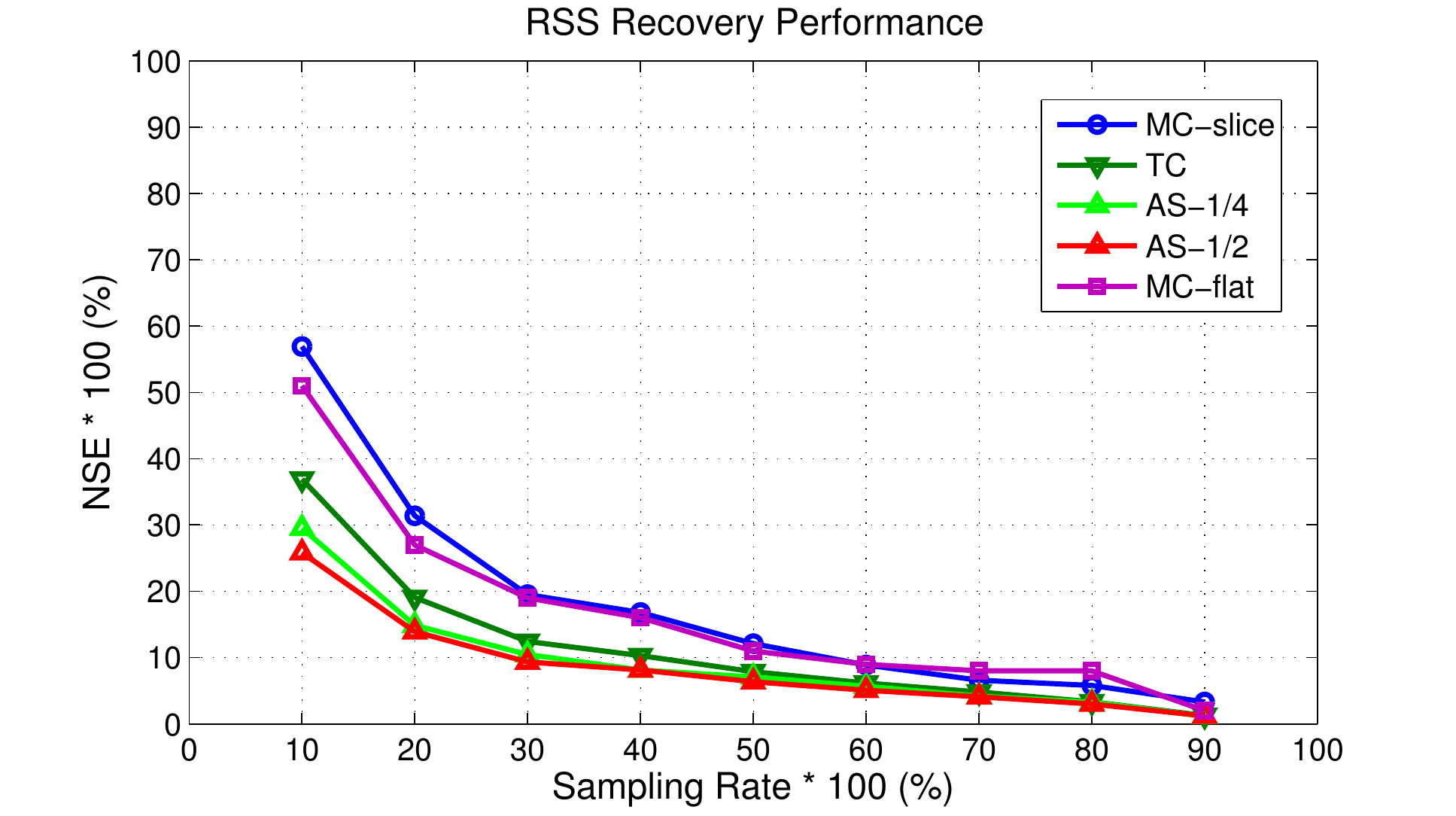}
  \caption{Tensor recovery in the real-world data set. The same legend as in Fig. \ref{fig:recovery}}\label{fig:real_data_recovery}
  \end{figure}

\subsection{Localization Performance}

  An important factor that influences the localization error is the measurement noise. Note that for site survey, the engineer stays for a while to obtain a stable fingerprint. Therefore, we only consider measurement noise in the query fingerprint. The noise may come from the measuring process or the dynamics in the environment. High SNR and low SNR cases are considered. For the high SNR case we add $3$dBm Gaussian noise, while $10$dBm Gaussian noise is added for the low SNR case.

  We choose three representative localization techniques for comparison, namely, weighted KNN, the kernel approach, and support vector machine (SVM).
  KNN (weighted KNN) is the most widely used technique since it is simple and is reported to have good performance in indoor localization systems \cite{RADAR2000}\cite{Yang2012MobiCom}. The kernel approach is an improved scheme over weighted KNN, which can be regarded as the basic principle of all machine learning-based localization approaches. The kernel function encapsulates the complicated relationship between RF fingerprints and physical positions. SVM is an efficient machine learning method widely used in fingerprint-based indoor localization. SVM uses kernel functions and tries to learn the complicated relationship between RF fingerprints and physical positions by regression.

  In addition to localization based on the estimated tensor $\hat{\mathcal{T}}$, we also test the above three localization techniques on the samples $\mathcal{Y}$ (without doing any reconstruction or estimation). Let DL denote {\em direct localization (DL} on uniformly sampled fingerprints, while DL-$1/4$ and DL-$1/2$ denote direct localization on adaptively sampled fingerprints with allocation ratios $\delta = 1/4$ and $\delta = 1/2$, respectively.

\subsubsection{Weighted KNN}

\begin{figure}[t]
\subfigure{\includegraphics[width=0.49\textwidth]{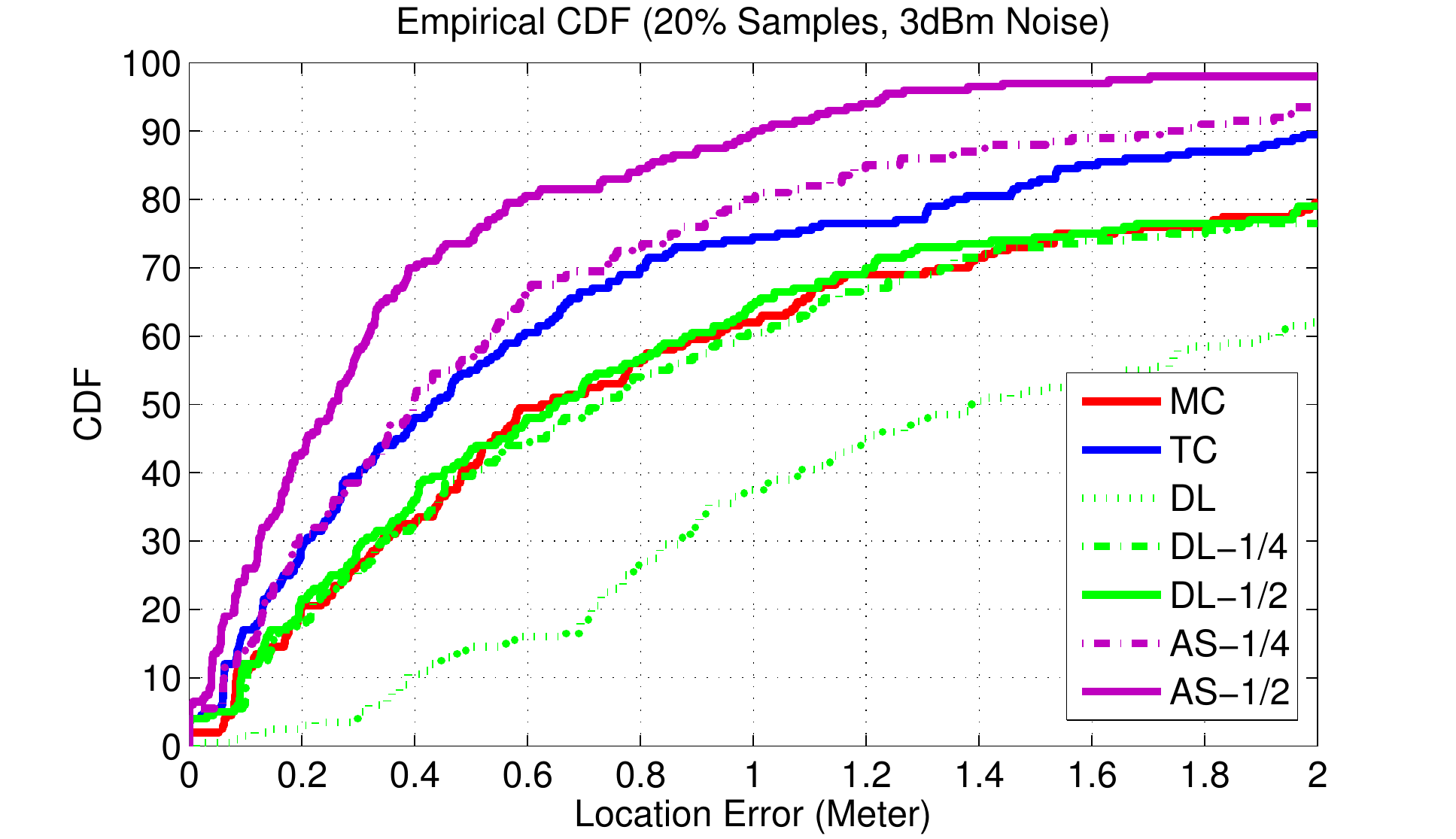}}
\subfigure{\includegraphics[width=0.49\textwidth]{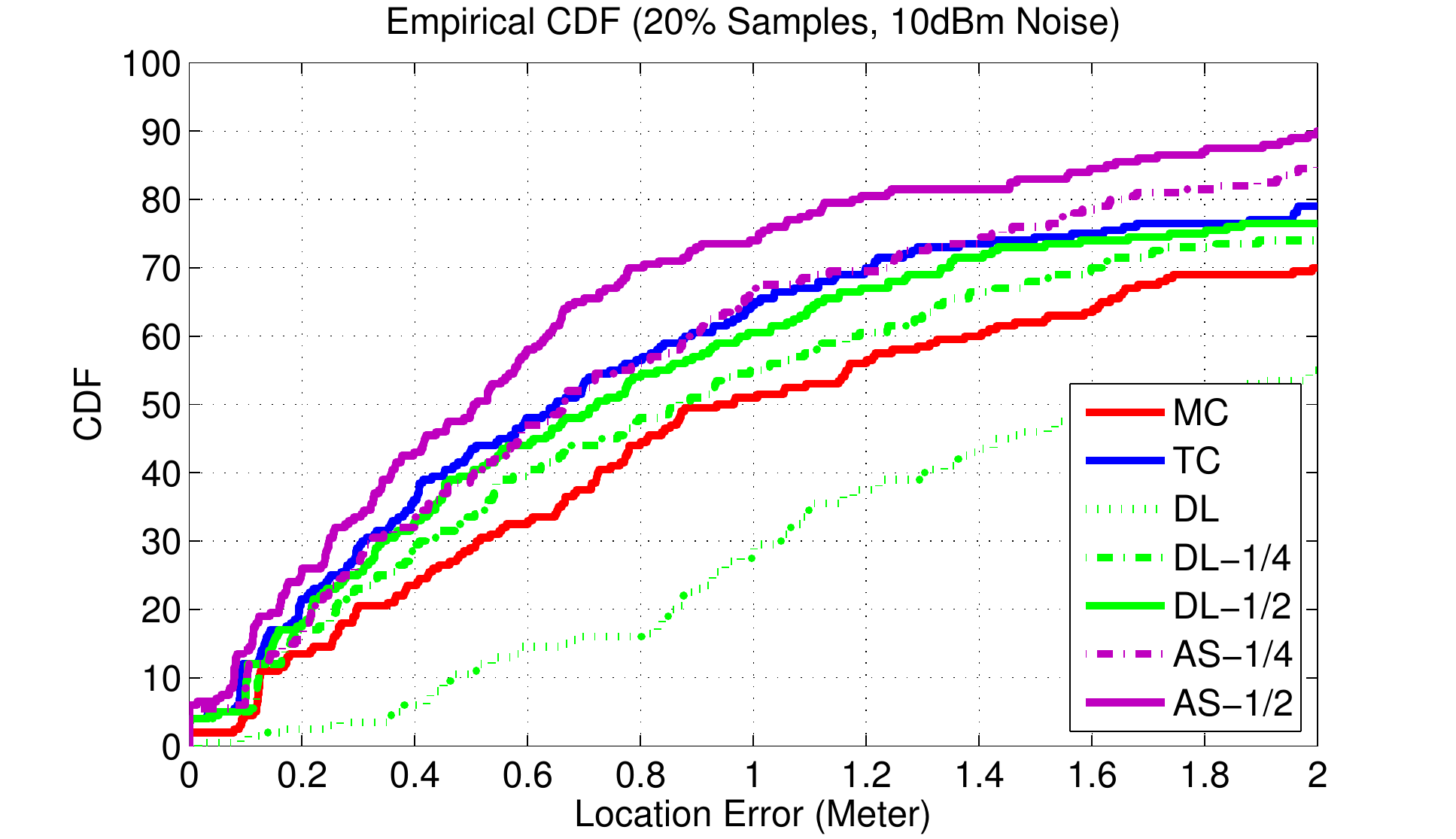}}
\caption{The empirical CDFs of localization error of KNN with $20\%$ samples for $3$dBm noise case and $10$dBm noise case. MC denotes matrix completion via uniform sampling, TC denotes tensor completion via uniform sampling, DL denotes direct localization on uniformly sampled fingerprints, while AS denotes our adaptive sampling scheme. For adaptive DL and AS, there are two allocation ratios $\delta = 1/4$ and $\delta = 1/2$.}\label{fig:KNN}
\end{figure}

  Let $k$ be a fixed positive integer which are usually set to be $1,3,5,7,$ etc., consider a sampled fingerprint $r_p$ at reference point $p$. Find within the fingerprint database $\mathcal{F}$ the reference point locations $p_1,p_2,...,p_k$ whose fingerprints are nearest to $r_p$. Then, estimate the location $p$ by weighted averaging $p_1,p_2,...,p_k$ as follows:
   \begin{equation}\label{knn_form}
   \hat{p} = 
   \sum_{i=1}^{k}p_i~ \frac{ \frac{1}{d(r_{p_i},r_p) + d_0} }{\sum_{i=1}^{k} \frac{1}{d(r_{p_i},r_p) + d_0}},
   \end{equation}
   where $d(r_{p_i},r_p)$ is the Euclidean distance between the two fingerprints, and $d_0$ is a small real constant used to avoid division by zero. In the simulations, we find that $k=5$ and $d_0 = 0.01$dBm are the best since the reference points in our experiments are located on a grid map with each reference point having $4$ neighbors (plus itself leads to $k=5$). Note that we estimate the x-coordinate and y-coordinate separately.

   Fig. \ref{fig:KNN} shows the empirical CDFs of localization error of weighted KNN with $20\%$ samples for $3$ dBm noise and $10$ dBm noise. Besides the recovery schemes  TC and MC,  we also consider direct localization (DL) with samples drawn uniformly from the grid map and those output by the adaptive sampling approach. We find that the proposed adaptive sampling approach (both $\delta = 1/4$ and $\delta = 1/2$) dramatically outperforms other schemes. In the high SNR case, $97\%$ ($90\%$ respectively) of the reference points can be localized with error less than $2$m ($1$m), while $90\%$ ($95\%$ respectively) in the low SNR case. Such performance improvements may be explained by the performance of direct localization (uniform, adaptive with $\delta = 1/4$ and $\delta = 1/2$). Since in both high and low SNR cases, adaptive sampled RF fingerprints are more useful for localization than non-adaptive sampled RF fingerprints. To some extent, it shows that our adaptive sampling approach has identified entries that are highly informative.

\subsubsection{Kernel Approach}
\begin{figure}[t]
\subfigure{\includegraphics[width=0.5\textwidth]{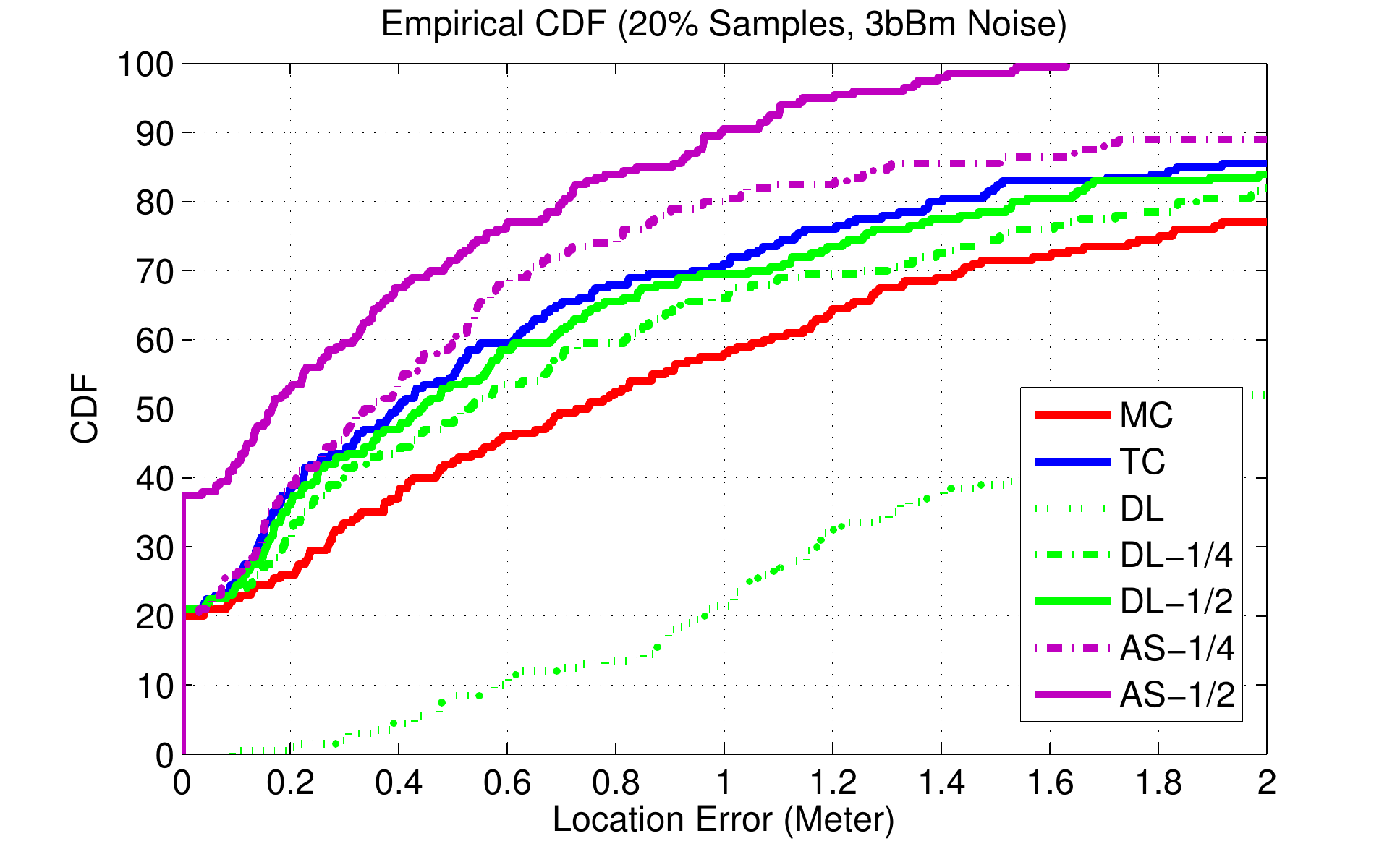}}
\subfigure{\includegraphics[width=0.5\textwidth]{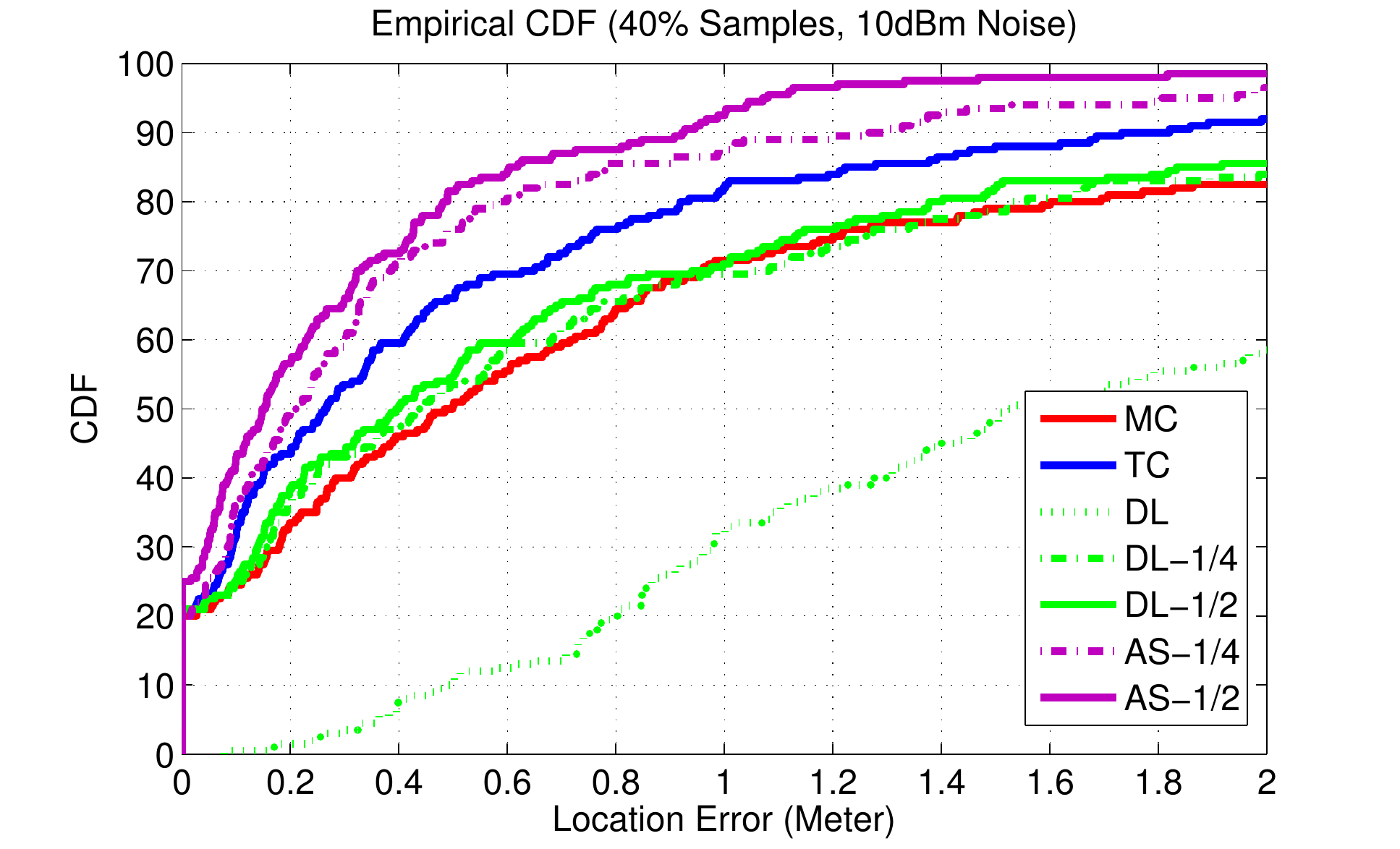}}
\caption{The empirical CDFs for the kernel approach with the same legend as Fig. \ref{fig:KNN}.}\label{fig:kernel}
\end{figure}

   We estimate the coordinates as $\hat{p}  = \sum_{i \in H(p)} \phi(r_{p_i}, r_p) p_i + \alpha_0$,
   where $H(p)$ denotes the set of reference points whose fingerprints have the smallest Euclidean distances to the query fingerrpint, $\phi(\cdot)$ is the kernel function which can be polynomial functions, Gaussian functions, and exponential functions. In the simulations, we set $\phi(\cdot)$ to be the square distance, the cardinality of $H(p)$ be $50$. Therefore, we have:
   \begin{equation}\label{kernel_location}
   \hat{p}  = \sum_{i \in H(p)} p_i \frac{ \frac{1}{d^2(r_{p_i},r_p) + d_0^2} }{\sum_{i=1}^{k} \frac{1}{d^2(r_{p_i},r_p) + d_0^2}}.
   \end{equation}
   We can see the strong similarity with the weighted KNN. A major difference is that the number of reference points used by the kernel method is much larger than the number of neighboring reference points used by weighted KNN.

   Fig. \ref{fig:kernel} shows the empirical CDFs of localization error of the kernel approach with $20\%$ samples for $3$dBm and $10$dBm noise. Similar to KNN, the proposed scheme outperforms all other schemes. There are two interesting things to notice. First, for the kernel approach, $38\%$ testing points has zero error for adaptive sampling with $\delta = 1/2$ and $3$dBm noise. Second, compared with Fig. \ref{fig:KNN}, the kernel approach is more robust to noise. Surprisingly, the kernel approach performs better in the $10$dBm noise case than in the $3$dBm noise case. The reason is that with lower SNR, the kernel function captures the location centroid of WiFi access points, i.e., the coordinates of access points. Because (\ref{kernel_location}) is dominant by nearby (relative to the user's location) access points, and with $50$ ($|H(p)|=50$) reference points rather than $k=5$ reference points, (\ref{kernel_location}) essentially captures the location of the dominating access points \cite{Kuo2011TMC}.

\subsubsection{SVM Approach}

   Support vector machines (SVM) separate reference points using linear hyperplanes in the fingerprint space. The x-coordinate and y-coordinate are estimated separately.
   Let $(x_j,y_j)$ denote the coordinates of reference point $p_j$ and $r_{p_j}$ denote the fingerprint. For $x_j$, the training data $\{(x_j, y_j),r_{p_j}\}$ become labeled-pairs $\{\Delta_i,r_{p_i}\}$ where the input is $r_{p_i}$, and the output is $\Delta_i =1$ if $x_i = x_j $ and $\Delta_i =-1$ if $x_i \neq x_j $.
   SVM construct a classifier of the form $\Delta_i(r_{p_i}) = \text{sign}[\sum_{i=1}^{N_1N_2} w_i \Delta_i \phi(r_{p_i},r_{p_j}) + b]$, where $w_i$ are positive real constants and $b$ is a real constant. The kernel function $\phi(\cdot,\cdot)$ typically has the following choices: $\phi(r_{p_i},r_{p_j}) = r_{p_i}^T r_{p_j}$ (linear SVM); $\phi(r_{p_i},r_{p_j}) = (r_{p_i}^T r_{p_j} + 1)^d$ (polynomial SVM of degree $d$); $\phi(r_{p_i},r_{p_j}) = \exp\{ - ||r_{p_i} - r_{p_j}||_2^2 /\sigma^2\}$ (radius basis function kernel SVM); $\phi(r_{p_i},r_{p_j}) = \tanh[\kappa r_{p_i}^T r_{p_j} + \theta]$ (two-layer neural SVM), where $\sigma, \kappa, \theta$ are constants. We use linear SVM in the experiments, and according to the structural risk minimization principle \cite{SVM}, we seek the solution of the following optimization problem: $\min_{w,b,\xi} ~\frac{1}{2}||w||_2 + C \sum_{i=1}^{N} \xi_i$ such that
   \begin{equation}
\left\{
   \begin{aligned}
   w^T \cdot r_{p_i} + b &\geq 1 - \xi_i,~~\text{if}~~x_j = x_i, \\
   w^T \cdot r_{p_i} + b &\leq-1 + \xi_i,~~\text{if}~~x_j \neq x_i,\\
   \xi_i &~> 0.~~\\
   \end{aligned}
   \right.
   \end{equation}
   where $C$ is the penalty parameter of the error term and $\xi_i$ are introduced in case a separating hyperplane in this higher dimensional space does not exist. We use the SVM library \cite{Lib_svm}.

   Fig. \ref{fig:SVM} shows the empirical CDFs of localization error of the SVM approach with $20\%$ samples for $3$dBm and $10$ dBm noise. We can see that the performance of SVM is quite similar to the kernel approach in Fig. \ref{fig:kernel}. Notice that for adaptive sampling, the allocation ratio $\delta$ does not affect the localization error that much as in the KNN approach and the kernel approach. It seems that for high SNR case, we should allocate more sampling budget to the $1$st-pass sampling, while more sampling budget to the $2$nd-pass sampling for low SNR case. Since for high SNR, the $1$st-pass sampling is able to better locate more informative columns of $\mathcal{G}$ (essentially those more informative columns lies near WiFi access points), while for low SNR it is better to have an accurate estimation of the low-dimensional subspace which is the aim of the $2$nd-pass sampling.

\begin{figure}[t]
\subfigure{\includegraphics[width=0.48\textwidth]{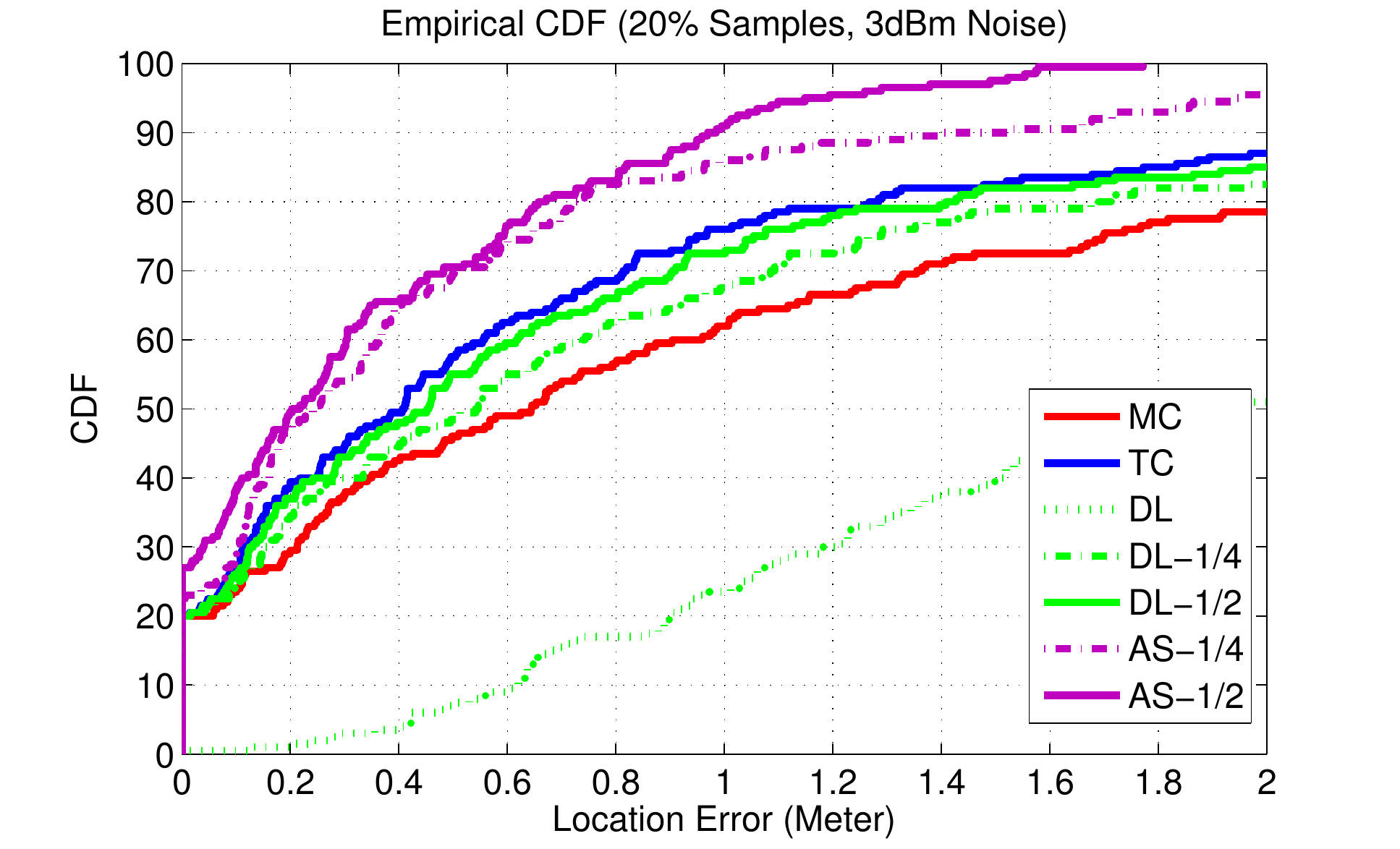}}
\subfigure{\includegraphics[width=0.48\textwidth]{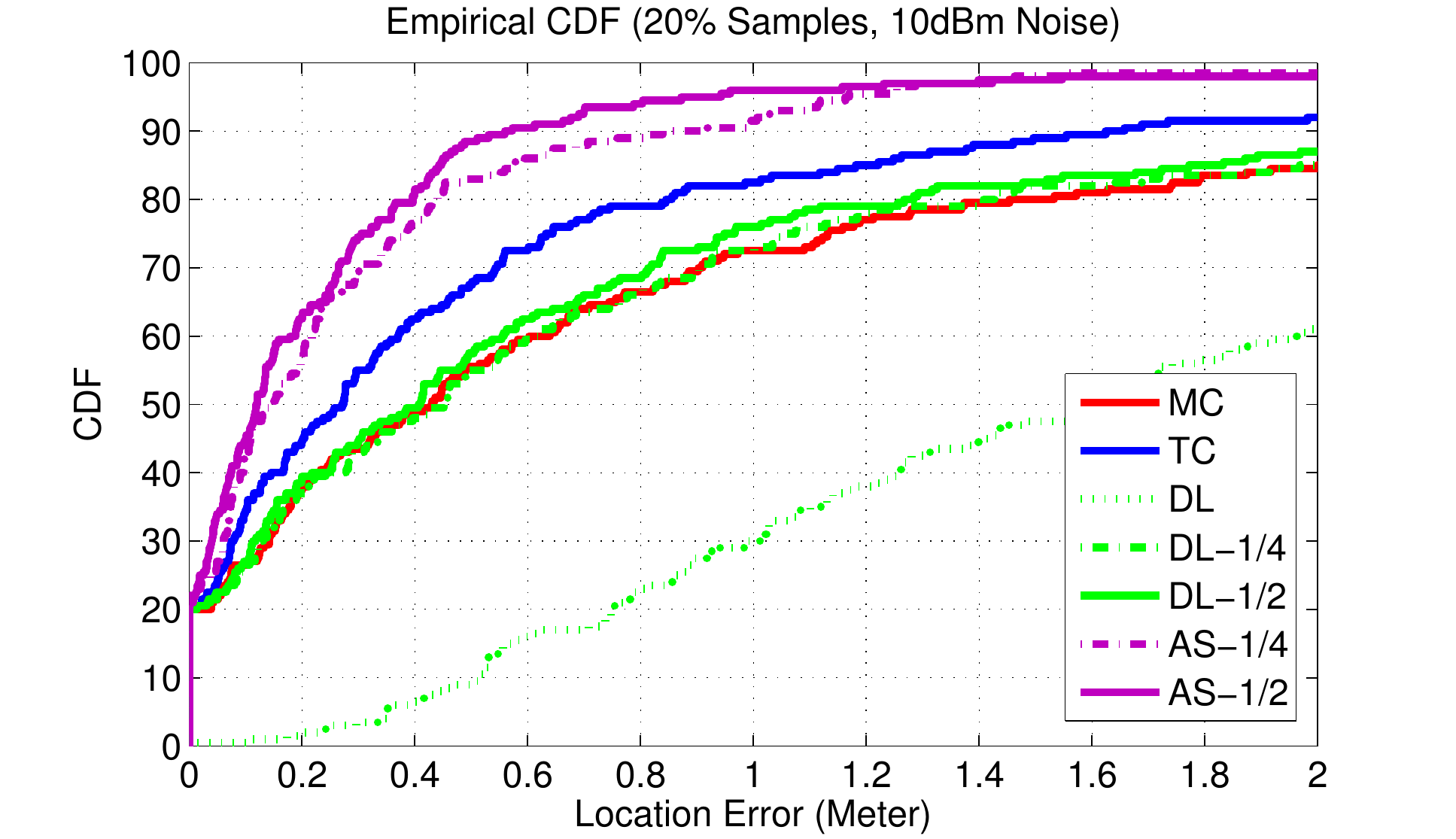}}
\caption{The empirical CDFs for SVM with the same legend as Fig. \ref{fig:KNN}.}\label{fig:SVM}
\end{figure}

\subsubsection{Reduction of Sampled Reference Points}
\begin{figure}[t]
\centering
\includegraphics[width=0.48\textwidth]{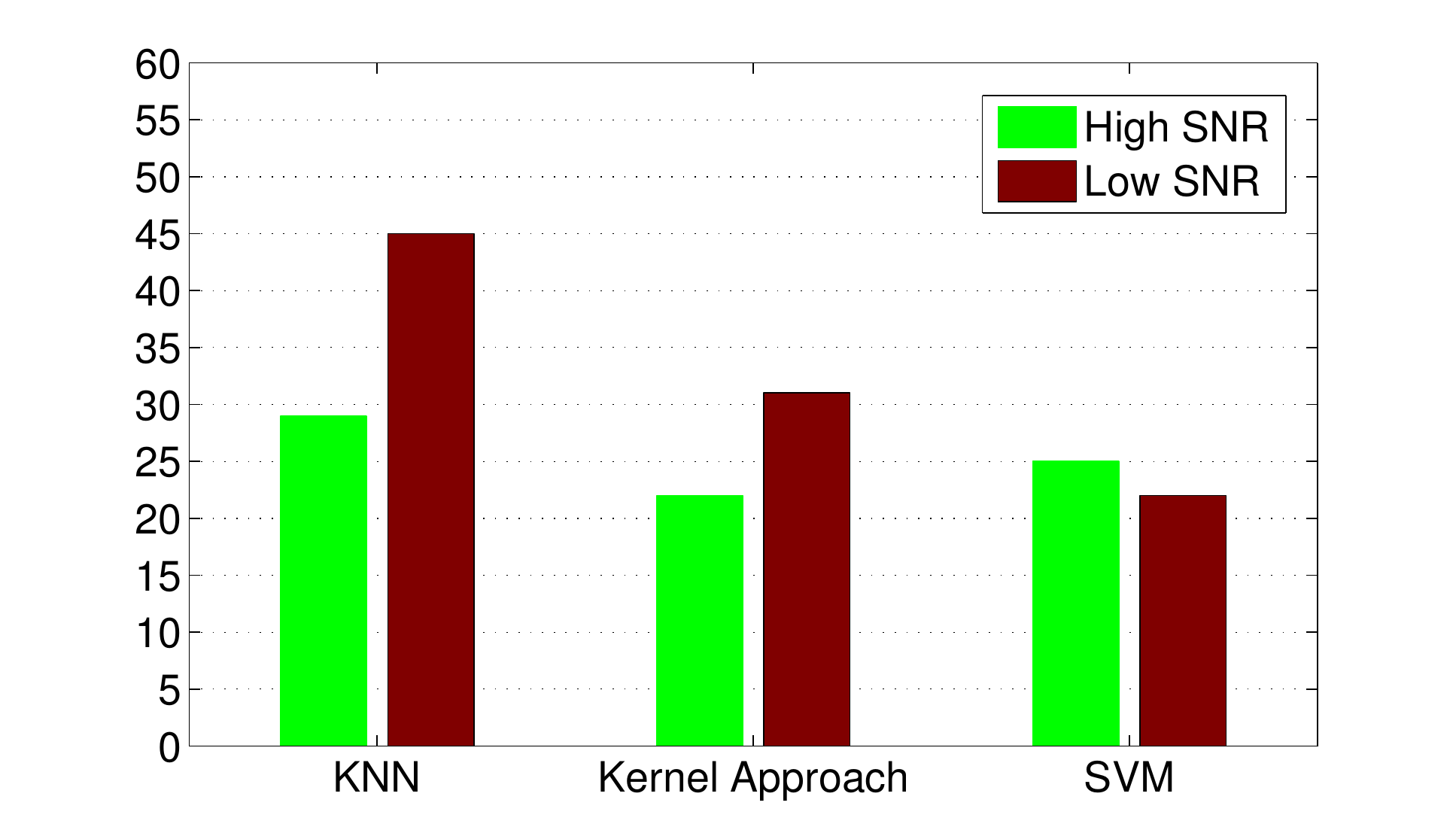}
\caption{The sampling budget needed for KNN, kernel approach and SVM.}\label{fig:budget}
\end{figure}

   We are interested in cutting down the sampling budget. We apply the adaptive sampling scheme with $\delta = 1/2$ and then run recovery-and-localization experiment. Through $20$ simulations, keeping these with $[94\% \sim 96\%]$-percentile localization error being $1$m, and calculate the average sampling rate. Fig. \ref{fig:budget} shows the results for both the high and low SNR cases. For high SNR, the kernel approach needs the least sampling budget which is $23\%$, while SVM needs $22\%$ for low SNR. Maintaining similar localization accuracy of KNN, the amount of samples required by the adaptive sampling approach is cut down by $71\% $ for the high SNR case and $55\%$ for the low SNR case.

%% file: Sections/Bib.tex
\bibliographystyle{IEEETran}

%% file: Sections/Proof.tex
\appendix

\subsection{Proof of Lemma \ref{lemma_singleround}}

   \begin{IEEEproof}
   We will construct $r$ lateral slices $w^{(1)},..., w^{(r)} \in \mathcal{U}$, 
   organized as $\mathcal{W} \in \mathbb{R}^{N_1 \times r \times N_3}$, 
   and use $\mathcal{W}$ to upper bound the projection because
   $||\mathcal{T} - P_{\mathcal{U} \cup \text{t-span}(\mathcal{S}), r}(\mathcal{T})||_F^2 \leq ||\mathcal{T} - P_{\mathcal{W}}(\mathcal{T})||_F^2$,
   so we work with $\mathcal{W}$ in the following.
   For each $j=1,...,N_2$ and for each $l=1,...,s$, we define random variables:
   \begin{equation}
   X_l^{(i)} = \frac{1}{\hat{p}_j} \mathcal{E}(:,j,:) \ast \mathcal{V}^T(i,j,:) ~\text{with probability~} \hat{p}_j,
   \end{equation}
   where $i \in \{1,2,...,N_1 \}$. That is the t-product of the $j$-th lateral slice of $\mathcal{E}$ and the $(i,j)$-th mode-3 tube of $\mathcal{V}^T$, scaled by the sampling probability. Defining $X^{(i)} = \frac{1}{s}\sum\limits_{l=1}^{s} X_l^{(i)}$, we have:
   \begin{equation}
   \mathbb{E}[X^{(i)}] = \mathbb{E}[X_l^{(i)}] = \sum\limits_{j=1}^{N_2} \frac{\hat{p}_j}{\hat{p}_j} \mathcal{E}(:,j,:) \ast \mathcal{V}^T(i,j,:) = \mathcal{E} \ast \mathcal{V}(:,i,:).
   \end{equation}
   Defining $w^{(i)} = P_{\mathcal{U}}(\mathcal{T}) \ast \mathcal{V}(:,i,:) + X^{(i)} \in \mathbb{R}^{N_1 \times 1 \times N_3}$ and using the definition of $\mathcal{E}$, it is easy to have:
   \begin{eqnarray}
   \mathbb{E}[w^{(i)}] &=& \mathbb{E}[  (\mathcal{T}-\mathcal{E})\ast \mathcal{V}(:,i,:) + X^{(i)}]
   =  \mathcal{U}(:,i,:) \ast \sigma_i \nonumber\\&&- \mathcal{E}\ast \mathcal{V}(:,i,:) + \mathbb{E}[ X^{(i)}] =\mathcal{U}(:,i,:) \ast \sigma_i, 
   \end{eqnarray}
   where $\sigma_i$ is the $i$-th diagonal tube of $\Theta$.
%

   Next, we proceed to bound the second central moment:
   \begin{eqnarray}
&&   w^{(i)} - \mathcal{U}(:,i,:) \ast \sigma_i = X^{(i)} - \mathcal{E}\ast \mathcal{V}(:,i,:), \nonumber\\
&&   \mathbb{E}[ || w^{(i)} - \mathcal{U}(:,i,:) \ast \sigma_i ||^2 ] =
   \mathbb{E}[||X^{(i)} - \mathcal{E}\ast \mathcal{V}(:,i,:) ||^2 ]\nonumber\\
&&   =\mathbb{E}[||X^{(i)}||_F^2] - || \mathcal{E} \ast \mathcal{V}(:,i,:) ||_F^2, \nonumber\\
&&   \mathbb{E}[ || X^{(i)} ||_F^2 ] = \frac{1}{s^2}\sum_{l=1}^{s} \mathbb{E}[ ||X_l^{(i)}||_F^2 ]  + \frac{s-1}{s} ||\mathcal{E} \ast \mathcal{V}(:,i,:) ||_F^2.\nonumber
   \end{eqnarray}

   Thus, the second central moment is:
   \begin{equation}
   \mathbb{E}[ || w^{(i)} - \mathcal{U}(:,i,:) \ast \sigma_i ||^2 ] = \frac{1}{s^2}\sum_{l=1}^{s} \mathbb{E}[ ||X_l^{(i)}||_F^2 ] - \frac{1}{s} || \mathcal{E} \ast \mathcal{V}(:,i,:) ||_F^2.
   \end{equation}

Now, we use the probability $\hat{p}_j$ to evaluate each term in the summation as follows.
\begin{eqnarray}
 &&  \mathbb{E}[ ||X_l^{(i)}||_F^2 ] = \sum_{j=1}^{N_2} \hat{p}_j \frac{ || \mathcal{E}(:,j,:) \ast \mathcal{V}^{T}(i,j,:) ||_F^2 }{\hat{p}_j^2} \nonumber\\&&\leq \sum_{j=1}^{N_2} \frac{1 + \alpha_2}{1 - \alpha_1} \frac{|| \mathcal{E}(:,j,:) \ast \mathcal{V}^{T}(i,j,:) ||_F^2}{|| \mathcal{E}(:,j,:)||_F^2} ||\mathcal{E}||_F^2.
   \end{eqnarray}
Now note that,
\begin{align}
&\frac{\| \mathcal{E}(:,j,:) \ast \mathcal{V}^{T}(i,j,:) \|_F^2}{\| \mathcal{E}(:,j,:)\|_F^2}  = \frac{\sum_{k=1}^{N3} \| \breve{\mc{E}}(:,j,k)\|_{F}^{2} |\breve{\mc{V}}^{T}(i,j,k)|^2}{\sum_{k=1}^{N3} \| \breve{\mc{E}}(:,j,k)\|_{F}^{2}}\nonumber\\
& \leq \frac{\left(\sum_{k=1}^{N3} \| \breve{\mc{E}}(:,j,k)\|_{F}^{2}\right) \max_{k}|\breve{\mc{V}}^{T}(i,j,k)|^2}{\sum_{k=1}^{N3} \| \breve{\mc{E}}(:,j,k)\|_{F}^{2}}.
\end{align}
Using the assumption that $\sum_{j=1}^{N_2} \max_{k}|\breve{\mc{V}}^{T}(i,j,k)|^2 \leq \xi_0$, we obtain,
\begin{align}
 \mathbb{E}[ ||X_l^{(i)}||_F^2 ] \leq \xi_0 \frac{1 + \alpha_2}{1 - \alpha_1} \|\mathcal{E}\|_F^2
 \end{align}

   Therefore, we have the following upper bound on the second central moment.
   \begin{equation}
   \mathbb{E}[ || w^{(i)} - \mathcal{U}(:,i,:) \ast \sigma_i ||^2 ] \leq \frac{1}{s} \cdot \xi_0 \frac{1 + \alpha_2}{1 - \alpha_1} ||\mathcal{E}||_F^2.
   \end{equation}

   To complete the proof, let $y^{(i)} = w^{(i)} / \sigma_i$ (note that here $/$ is the inverse of the t-product that $w^{(i)} = y^{(i)} \ast \sigma_i$), and define a tensor $\mathcal{Y} = \sum\limits_{j=1}^{k} y^{(i)} \ast \mathcal{U}(:,i,:)$ with $k > r$. Since $y^{(i)} \in \mathcal{W}$, the subspace of $\mathcal{Y}$ is contained in $\mathcal{W}$, so $||\mathcal{T} - P_{\mathcal{W}}(\mathcal{T})||_F^2 \leq ||\mathcal{T} - \mathcal{Y} ||_F^2$.
   \begin{eqnarray}
&&   ||\mathcal{T} - \mathcal{Y} ||_F^2 = \sum\limits_{i=1}^{d} || (\mathcal{T} - \mathcal{Y}) \ast \mathcal{V}^{T}(i,:,:)||_F^2\nonumber\\
&&    \leq \sum\limits_{i=k+1}^{d} ||\sigma_i||_F^2 + \sum\limits_{i=1}^{k} ||\mathcal{U}(:,i,:) \ast \sigma_i - w^{(i)}||_F^2.
   \end{eqnarray}

   We then use Markov's inequality on the second term. Specifically, with probability $\geq 1 - \rho$, we have:
   \begin{eqnarray}
  && ||\mathcal{T} - \mathcal{Y} ||_F^2  \nonumber\\&&\leq ||\mathcal{T} - \mathcal{T}_k ||_F^2 + \frac{1}{\rho} \mathbb{E} [\sum\limits_{i=1}^{k} ||\mathcal{U}(:,i,:) \ast \sigma_i - w^{(i)}||_F^2] \nonumber\\&&\leq ||\mathcal{T} - \mathcal{T}_r ||_F^2 + \frac{r}{\rho s} \cdot \xi_0 \frac{1 + \alpha_2}{1- \alpha_1} ||\mathcal{E}||_F^2.
   \end{eqnarray}
   \end{IEEEproof}

\subsection{Some Lemmas that will be used for other results in Appendix}

   We need the following two versions of Bernstein's inequalities in our proofs.

   \begin{lemma}
   (Scalar Version) Let $X_1,...,X_n$ be independent centered scalar variables with $\sigma^2 = \sum_{i=1}^{n} \mathbb{E}[X_i^2]$ and $R = \max_{i} |X_i|$. Then
   \begin{equation}
   \mathbb{P}\left( \sum\limits_{i=1}^{n} X_k \geq t\right) \leq \exp\left( \frac{-t^2}{2\sigma^2 + \frac{2}{3}Rt} \right).
   \end{equation}
   \end{lemma}

   \begin{lemma}
   (Vector Version) Let $X_1,...,X_n$ be independent centered random vectors with $\sum_{i=1}^{n} \mathbb{E}||X_i||_2^2 \leq V$. Then for any $t \leq V(\max_{i} ||X_i||_2)^{-1}$.
   \begin{equation}
   \mathbb{P}\left( ||\sum\limits_{i=1}^{n} X_k ||_2 \geq \sqrt{V}+ t\right) \leq \exp\left( \frac{-t^2}{4V} \right).
   \end{equation}
   \end{lemma}

   \begin{lemma}\label{lemma_p1}
   Let $\mathcal{Z} =  \mathcal{P}_{\mathcal{U}^{\bot}} (\mathcal{T}(:,j,:))$, $\mu(\mathcal{Z}) = \frac{N_1 ||\mathcal{Z}||_{F\infty}^2}{ ||\mathcal{Z}||_F^2}$,
where $ ||\mathcal{Z}||_{F\infty} = \max_{i} \| \mathcal{Z}(i,1,:)\|_{F},\,\, i \in \{1,2,...,N_1\}$, and
    $$\alpha = \sqrt{2\frac{\mu(\mathcal{Z})}{m}\log(1/\rho)} + \frac{2\mu(\mathcal{Z})}{3m}\log(1/\rho), $$ then with probability $\geq 1- 2 \rho$,
   \begin{equation}
   (1- \alpha) \frac{m}{N_1} || \mathcal{Z} ||_F^2 \leq ||\mathcal{Z}(\Omega_j^1,:) ||_F^2 \leq ( 1 + \alpha) \frac{m}{N_1} || \mathcal{Z} ||_F^2.
   \end{equation}
   \end{lemma}

   \begin{proof}
   This proof is an direct application of the scalar version of Bernstein's inequality. The left orthogonal tensor $\mathcal{U}$ is of size $N_1 \times r \times N_3$. Let $T(:,\Omega_j^1,:)$ denote the sampled fingerprints in $j$-th column of $\mathcal{G}$ during the $1$st-pass sampling, and $\mathcal{Z} =  \mathcal{P}_{\mathcal{U}^{\bot}} (\mathcal{T}(:,j,:)) \in \mathbb{R}^{N_1 \times 1 \times N_3}$ is its projection onto $\mathcal{U}^{\bot}$. Where $\Omega_j^1$ corresponds to those sampled $m$ reference points, and $\Omega_j^1(i)$ is the $i$-th element in set $\Omega_j^1$.

   Define a random variable $$X_i = ||\mathcal{Z}(\Omega_j^1(i),1,:)||_2^2 - \frac{1}{N_1}||\mathcal{Z}||_F^2\, ,$$ so that $\mathbb{E}[X_i] = \frac{1}{N_1}\sum_{i=1}^{N_1} ||\mathcal{Z}(i,1,:)||_2^2  - \frac{1}{N_1}||\mathcal{Z}||_F^2 =0$, and $\sum_{i=1}^{m} X_i = ||\mathcal{Z}(\Omega_j^1,:)||_F^2 - \frac{m}{N_1}||\mathcal{Z}||_F^2$. We compute the variance and bound for $X_i$ as follows.
   \begin{equation}
   \begin{split}
   \phi^2 &= \sum_{i=1}^{m} \mathbb{E}[X_i^2] \leq \frac{m}{N_1} \sum\limits_{i=1}^{N_1} ||\mathcal{Z}(i,1,:)||_2^4 \leq \frac{m}{N_1}||\mathcal{Z}||_{F\infty}^2 ||\mathcal{Z}||_F^2, \\
   \Phi &= \max_{i \in \{1,2,...,N_1\}} |X_i| \leq ||\mathcal{Z}||_{F\infty}^2,
   \end{split}
   \end{equation}

   Applying the Berstein's inequality:
   \begin{equation}
   \mathbb{P} \left(\left| \sum\limits_{i=1}^{m} X_i \right| > t\right) \leq 2 \exp\left( \frac{-t^2}{2||\mathcal{Z}||_{F\infty}^2(\frac{m}{N_1} ||\mathcal{Z}||_F^2 + \frac{1}{3} t)} \right).
   \end{equation}

   Using the definition of $\mu(\mathcal{Z})$ and setting $t=\alpha \frac{m}{N_1} || \mathcal{Z} ||_F^2$, the above bound becomes:
   \begin{equation}
   \mathbb{P} \left(\left| \sum\limits_{i=1}^{m} X_i \right| > \alpha \frac{m}{N_1} || \mathcal{Z} ||_F^2 \right) \leq 2 \exp\left( \frac{-\alpha^2 m}{2\mu(\mathcal{Z}) (1 + \alpha/3)} \right).
   \end{equation}
   Plugging in the definition of $\alpha$, then the probability $\leq 2\rho$.
   \end{proof}

   \begin{lemma}\label{lemma_p2}
   Let $\beta = 1 + 2 \sqrt{\log(1/\rho)}$, with probability $\geq 1 - \rho$,
   \begin{equation}
   || \mathcal{U}_{\Omega_j^1}^T \ast \mathcal{Z}(\Omega_j^1,:) ||_F^2 \leq \beta \frac{m}{N_1} \cdot \frac{r \mu(\mathcal{U})}{N_1} ||\mathcal{Z}||_F^2.
   \end{equation}
   \end{lemma}
   \begin{proof}
   The proof is an application of the vector version of Bernstein's inequality.
   Let  $X_i = \mathcal{U}^T(:,\Omega_j^1(i),:) \ast \mathcal{Z}(\Omega_j^1(i),1,:)$. Since 
   $\mathcal{Z} =  \mathcal{P}_{\mathcal{U}^{\bot}} (\mathcal{T}(:,j,:))$, the expectation and variance are:
   \begin{eqnarray}
  && \mathbb{E}[ X_{i}] = \frac{1}{N_1}\sum\limits_{k=1}^{ N_1} \mathcal{U}^T(:,k,:) \ast \mathcal{Z}(k,1,:) = 0_{1 \times N_3},\nonumber \\
&&   \sum\limits_{i=1}^{m} \mathbb{E}||X_i||_2^2 = \frac{m}{N_1} \sum\limits_{k=1}^{N_1} || \mathcal{U}^T(:,k,:) \ast \mathcal{Z}(k,1,:) ||_F^2 \nonumber\\&&= \frac{m}{N_1} \sum\limits_{k=1}^{N_1} \sum_{j=1}^{N_3} \frac{1}{N_3}\|\breve{\mathcal{U}}^T(:,k,j) \breve{\mathcal{Z}}(k,1,j)\|_2^2\nonumber\\
   &&\leq \frac{m}{N_1} \left(\frac{1}{N_3}\sum\limits_{k=1}^{N_1} \|\breve{\mathcal{U}}^T(:,k,:)\|_{F}^{2}\right) \left(  \sum\limits_{k=1}^{N_1}   \sum\limits_{j=1}^{N_3} ||\breve{\mathcal{Z}}(k,1,j)||_2^2 \right)\nonumber\\ 
   &&\leq \frac{m}{N_1} \frac{r\mu(\mathcal{U})}{N_1} ||\breve{\mathcal{Z}}||_F^2 \triangleq V.
   \end{eqnarray}

   Applying Bernstein's inequality, we have that with probability at least $1-\rho$:
   \begin{eqnarray}
&&   || \mathcal{U}_{\Omega_j^1}^T \ast \mathcal{Z}(\Omega_j^1,:) ||_F^2 \leq \sqrt{V}+\sqrt{4V\log(1/\rho)} \nonumber\\&&= \sqrt{\frac{m}{N_1} \frac{r\mu(\mathcal{U})}{N_1}} ||\breve{\mathcal{Z}}||_F (1 + 2 \sqrt{\log(1/\rho)}),
   \end{eqnarray}
   as long as:
   \begin{equation}
   t = \sqrt{4V \log(1/\rho)} \leq V(\max_{i} ||X_i||_2)^{-1}.
   \end{equation}
   Since $\max_{i} ||X_i||_2 \leq ||\breve{\mathcal{Z}}||_{F\infty} \sqrt{r\mu(\mathcal{U})/N_1}$ and using the incoherence assumption on $\mathcal{Z}$, i.e. $$ \sqrt{\mu(\mathcal{Z})}  = \frac{\sqrt{N_1} \mathcal{Z}_{F,\infty}}{\|\mathcal{Z}\|_{F}} = \frac{\sqrt{N_1} \|\breve{\mathcal{Z}}\|_{F,\infty}}{\|\breve{\mathcal{Z}}\|_{F}}$$,
   this condition translates to $m \geq 4 \mu(\mathcal{Z}) \log(1/\rho)$. Squaring the above inequality proves the lemma.
   \end{proof}

   \begin{lemma}\label{lemma:inverse}
   \cite{Recht2010} Let $\rho > 0$ and $m = \frac{8}{3} r\mu_0 \log(2r /\rho)$, then for any orthonormal matrix $U$,
   \begin{equation}
   || ({U}_{\Omega_j^1}^T {U}_{\Omega_j^1})^{-1}||_{op}^{2} \leq \frac{N_1}{(1- \gamma)m}
   \end{equation}
   with probability $\geq 1 - \rho$, provided that $\gamma < 1$. 
   \end{lemma}

   \begin{lemma}\label{space_detection}
   Let $\mathcal{U} \in \mathbb{R}^{N_1 \times r \times N_3}$ be an $r$-dimensional tensor-column subspace and $\mathcal{T}(:,j,:) = \mathcal{X} + \mathcal{Z}$ where $\mathcal{X} = \mathcal{P}_{\mathcal{U}} (\mathcal{T}(:,j,:))$ and $\mathcal{Z} =  \mathcal{P}_{\mathcal{U}^{\bot}} (\mathcal{T}(:,j,:))$. 
   Then with probability at least $1-4 \rho$:
   \begin{eqnarray}
   &&\frac{m(1-\alpha) - r\mu(\mathcal{U}) \frac{\beta}{1 - \gamma}}{N_1} ||\mathcal{Z}||_F^2 \leq ||\mathcal{T}(\Omega_j^1,j,:)\nonumber\\&& - \mathcal{P}_{U_{\Omega_j^1}} (\mathcal{T}(\Omega_j^1,j,:)) ||_F^2 \leq (1+\alpha) \frac{m}{N_1} ||\mathcal{Z}||_F^2,
   \end{eqnarray}
   where $m = |\Omega_j^1|$, $\alpha = \sqrt{2\frac{\mu(\mathcal{Z})}{m}\log(1/\rho)} + \frac{2\mu(\mathcal{Z})}{3m}\log(1/\rho)$, $\beta = (1 + 2 \sqrt{\log(1/\rho)})^2$, $\gamma = \sqrt{\frac{8r\mu(\mathcal{U})}{3m}\log(2r/\rho)}$, $\mu(\mathcal{Z}) = N_1||\mathcal{Z}||_{F\infty}^2 / ||\mathcal{Z}||_F^2$, and $ ||\mathcal{Z}||_{F\infty}$ denotes the tube with maximum $\ell_2$-norm.
   \end{lemma}

   \begin{IEEEproof}
   We begin by considering the evaluation of $\|\mathcal{T}(\Omega_j^1,j,:) - \mathcal{P}_{U_{\Omega_j^1}} (\mathcal{T}(\Omega_j^1,j,:)) \|_F^2$ in the Fourier domain. Specifically note that for any tensor $\mc{T}$, $\| \mathcal{T}\|_{F}^{2} = \frac{1}{N_3} \| \breve{\mathcal{T}} \|_{F}^{2} = \frac{1}{N_3} \sum_{k =1}^{N_3} \| \breve{\mathcal{T}}^{(k)} \|_{F}^{2} $ where $\breve{\mathcal{T}} = {\tt fft}(\mc{T},[\,],3)$ denoted the (un-normalized) Fourier transform of $\mc{T}$ along the third dimension and $\breve{\mc{T}}^{(k)}$ denotes the $k$-th frontal face of $\breve{\mc{T}}$.\\
Now note that from Remark \ref{computation}, it follows that, $\|\mathcal{T}(\Omega_j^1,j,:) - \mathcal{P}_{U_{\Omega_j^1}} (\mathcal{T}(\Omega_j^1,j,:)) \|_F^2 = \frac{1}{N_3} \sum_{k=1}^{N_3} \|\breve{\mc{T}}(\Omega_j^1,j,k) - \mathcal{P}_{\breve{\mc{U}}_{\Omega_j^1}^{(k)}} (\breve{\mc{T}}(\Omega_j^1,j,k)) \|_{F}^{2}$. Now each of the terms,
$\|\breve{\mc{T}}(\Omega_j^1,j,k) - \mathcal{P}_{\breve{\mc{U}}_{\Omega_j^1}^{(k)}} (\breve{\mc{T}}(\Omega_j^1,j,k)) \|_{F}^{2} = \|\breve{\mc{Z}}_{\Omega_j^1}^{(k)}\|_{F}^{2} - \|\breve{\mc{Z}}_{\Omega_j^1}^{(k)T} \breve{\mc{U}}_{\Omega_j^1}^{(k)}(\breve{\mc{U}}_{\Omega_j^1}^{(k)T} \breve{\mc{U}}_{\Omega_j^1}^{(k)})^{-1} \breve{\mc{U}}_{\Omega_j^1}^T \mathcal{Z}_{\Omega_j^1}^{(k)}\|_{F}\,\,,
$ where $\breve{\mc{Z}}_{\Omega_j^1}^{(k)} = \breve{\mc{Z}}(\Omega_j^1,1,k)$ is a column vector of size $|\Omega_j^1|$. Now note that,
$\|\breve{\mc{Z}}_{\Omega_j^1}^{(k)T} \breve{\mc{U}}_{\Omega_j^1}^{(k)}(\breve{\mc{U}}_{\Omega_j^1}^{(k)T} \breve{\mc{U}}_{\Omega_j^1}^{(k)})^{-1} \breve{\mc{U}}_{\Omega_j^1}^T \mathcal{Z}_{\Omega_j^1}^{(k)}\|_{F} \leq \|(\breve{\mc{U}}_{\Omega_j^1}^{(k)T} \breve{\mc{U}}_{\Omega_j^1}^{(k)})^{-1}\|_{op}^{2} \| \breve{\mc{U}}_{\Omega_j^1}^{(k)T} \mathcal{Z}_{\Omega_j^1}^{(k)} \|_{F}^{2}.
$ The proof then follows from the results of Lemma \ref{lemma_p1}, Lemma \ref{lemma_p2} and Lemma \ref{lemma:inverse}.

   \end{IEEEproof}

\subsection{Proof of Lemma \ref{lemma:sampling_budget}}\label{proof:sampling_budget}

   \begin{IEEEproof}
   The key is to apply Lemma \ref{space_detection} by taking a union bound across all rounds and all lateral slices. We first need to bound the incoherences, i.e., $\mu(\mathcal{Z})$, and $\mu(\hat{\mathcal{U}})$ in each round. As a direct consequence of Lemma 14 in \cite{Singh2013NIPS}, Lemma 15 in \cite{Singh2013NIPS} and Corollary 1 in \cite{Singh2015}, we know that with probability $1- \rho$, each lateral slice $\mathcal{T}(:,j,:)$ has incoherence $O(r\mu_0\log(1/\rho))$, and all the estimated tensor columns subspaces $\hat{\mathcal{U}}$ in the $2$nd-pass sampling have incoherence at most $O(\frac{r\mu_0}{Ls}\log(1/\rho))$.

   Take a union bound cross all lateral slices (columns of $\mathcal{G}$), while each $\rho$ term in Lemma \ref{space_detection} is replaced with $\rho/N_2$. Denote by $\hat{\mathcal{U}}$ the tensor column subspace projected onto during the $l$-th round of the $2$nd-pass sampling. Uniformly sample $m < N_1$ reference points in each column, the condition that $m \geq 8/3 dim(\hat{\mathcal{U}})\log(\frac{2rN_2}{\rho})$ ($dim(\cdot)$ means the $2$nd-dimension of the tensor-column subspace of $\hat{\mathcal{U}}$) is clearly satisfied, since $dim(\hat{\mathcal{U}}) \leq Ls = \frac{5L^2 N_1r}{2\rho \epsilon}$ and $\mu(\hat{\mathcal{U}}) \leq c \mu(\mathcal{U})\log(N_2/\rho)$.

   Next, we turn to analyze $\alpha, \beta, \gamma$. More specifically, we want $\alpha = O(1), \gamma = O(1)$ and $\frac{r\mu(\mathcal{U})}{m}\beta = O(1)$.

   For $\alpha$, $\alpha = O(1)$ implies that $m \geq C\mu(\mathcal{Z})\log(N_2/\rho) \geq Cr\mu_0\log^2(N_2/\rho)$. Therefore, by carefully choosing constant $C$ we can make $\alpha \leq 1/4$. For $\gamma$, $\gamma = O(1)$ implies that $m \geq C\mu(\mathcal{U})\log(rN_2/\rho) \geq Cr\mu_0\log(N_2/\rho)\log(rN_2/\rho)$. Therefore, by carefully choosing constant $C$ we can make $\gamma \leq 1/5$.
   For $\beta$, $\frac{k\mu(\mathcal{U})}{m}\beta = O(1)$ implies that $m \geq C r\mu(\mathcal{U}) \beta \geq C \frac{r^2\mu_0}{Ls}\log(N_2/\rho)(1 + 2 \sqrt{\log(N_2/\rho)})^2 = C\frac{r\mu_0 \rho}{L^2 N_2 \epsilon}\log^2(N_2/\rho)$.

   Therefore, $m \geq Cr\mu_0\log(N_2/\rho)\log(rN_2/\rho)$. Combining bounds on $\alpha, \beta, \gamma$, we get the probability estimation bound
   $\frac{2}{5} \cdot \frac{||\mathcal{E}(:,j,:)||_F^2}{||\mathcal{E}||_F^2} \leq \hat{p}_j \leq
   \frac{5}{2} \cdot \frac{||\mathcal{E}(:,j,:)||_F^2}{||\mathcal{E}||_F^2}.
   $

   The total number of reference points sampled in the $1$st-pass sampling is:
   \begin{equation}
   \delta M = N_2 m \geq Cr\mu_0N_2\log(N_2/\rho)\log(rN_2/\rho).
   \end{equation}
   We sample $s$ columns in $L$ rounds, the total number of reference points sampled in the $2$nd-pass sampling is
   \begin{equation}
   (1-\delta) M = Ls = CrN_2 (\lceil \log(N_1 N_2N_3)\rceil)^2 /(\rho \epsilon).
   \end{equation}
   Thus, we have $M = \delta M + (1-\delta) M \geq CrN_2 (\mu_0\log(N_2/\rho)\log(rN_2/\rho) + (\lceil \log(N_1 N_2N_3)\rceil)^2 /(\rho \epsilon))$.
   \end{IEEEproof}

\subsection{Proof of Lemma \ref{lemma:each_slice}}

   \begin{IEEEproof}
   Let $\mathcal{T}(:,j,:) = \mathcal{X} + \mathcal{Z}$ where $\mathcal{X} \in \hat{\mathcal{U}}$ and $\mathcal{Z} \in \hat{\mathcal{U}}^{\bot}$. So $\mathcal{X} = \hat{\mathcal{U}} \ast (\hat{\mathcal{U}}_{\Omega_j}^T \ast \hat{\mathcal{U}}_{\Omega_j})^{-1} \ast \hat{\mathcal{U}}_{\Omega_j}^T \ast \mathcal{X}(\Omega_j,j,:)$, and we are left with
   $|| \mathcal{Z} -   \hat{\mathcal{U}} \ast (\hat{\mathcal{U}}_{\Omega_j}^T \ast \hat{\mathcal{U}}_{\Omega_j})^{-1} \ast \hat{\mathcal{U}}_{\Omega_j}^T \ast \mathcal{Z} ||_F^2 = ||\mathcal{Z}||_F^2 + ||\hat{\mathcal{U}} \ast (\hat{\mathcal{U}}_{\Omega_j}^T \ast \hat{\mathcal{U}}_{\Omega_j})^{-1} \ast \hat{\mathcal{U}}_{\Omega_j}^T \ast \mathcal{Z}||_F^2$,
   where the cross term is zero because $\mathcal{Z} \in \hat{\mathcal{U}}^{\bot}$.

   Using Lemma \ref{lemma:inverse} (by again going into the Fourier domain) we can upper bound the second term by $ \frac{N_1^2}{(1- \gamma)^2m^2}$, while using Lemma \ref{lemma_p2} we can upper bound the first term by $\beta \frac{m}{N_1} \cdot \frac{r \mu(\mathcal{U})}{N_1} ||\mathcal{Z}||_F^2$. Combining these two yields the result.

   \end{IEEEproof}

\subsection{Proof of Theorem \ref{recovery_error}}

   \begin{IEEEproof}
   Note that $r \ll \min(N_1,N_2)$, and we set $\frac{(1-\delta)M}{N_2L} = \frac{5Lr}{2\rho \epsilon}$, Lemma \ref{lemma_ap} holds.
    We additionally assume that $\mathcal{T}'$ equals $\mathcal{T}$ on all of the unobserved entries as measurement noise on those entries do not effect our algorithm. That is to say, if $\Omega$ denotes the set of all sampled reference points over the course of the algorithm, tensor $\mathcal{N}$ has zero tubes on $\Omega^c$. We denote $\mathcal{T}' = \mathcal{T} + \mathcal{N}_{\Omega}$ where $\mathcal{R}(i,j,:)$ is a zero tube for $(i,j) \in \Omega$. We expand the Frobinus norm, $|| \hat{\mathcal{T}} - \mathcal{T}||_F \leq  || \hat{\mathcal{T}}- \mathcal{T}'||_F + ||\mathcal{T}' - \mathcal{T}||_F$,  and apply Lemma \ref{lemma_ap}, then we get:
   \begin{eqnarray}
   && || \hat{\mathcal{T}} - \mathcal{T}||_F^2 \leq 2 || \hat{\mathcal{T}}- \mathcal{T}'||_F^2 + 2||\mathcal{N}_{\Omega}||_F^2 \nonumber\\ && \leq \frac{5}{2} \left( \frac{1}{1-\epsilon} ||\mathcal{T}' - \mathcal{T}_r'||_F^2 + \epsilon^L ||\mathcal{T}'||_F^2 \right) + 2||\mathcal{N}_{\Omega}||_F^2.
   \end{eqnarray}
   According to Lemma \ref{best_r_rank}, $\mathcal{T}_r'$ is the best rank-$r$ approximation to $\mathcal{T}'$ and since $\mathcal{T}$ is known to have tubal-rank $r$, we have that $||\mathcal{T}' - \mathcal{T}_r'||_F \leq ||\mathcal{T}' - \mathcal{T}||_F$.
   Setting $L= \lceil \log_2(N_1N_2N_3) \rceil$, we get our final result:
   \begin{eqnarray}
   &&|| \hat{\mathcal{T}} - \mathcal{T} ||_F^2 \leq c_1 || \mathcal{T}' - \mathcal{T}||_F^2 + \frac{c_2}{N_1N_2N_3} ||\mathcal{T}'||_F^2 + c_3 ||\mathcal{N}_{\Omega}||_F^2 \nonumber\\&&\leq \frac{c'_2}{N_1N_2N_3} ||\mathcal{T}||_F^2 + (c_1 + c_3) ||\mathcal{N}_{\Omega}||_F^2.
   \end{eqnarray}
   Note that there exists a constant $c$ such that $||\mathcal{T}'||_F^2 \leq c ||\mathcal{T}||_F^2$.

   The inequality $||\mathcal{N}_{\Omega}||_F \leq \delta $ holds with high probability, for some $\delta >0$ and $\delta^2 \leq (MN_3 + \sqrt{8MN_3}) \sigma^2$ \cite{Candes2010}. Therefore, we have $||\mathcal{N}_{\Omega}||_F^2 \leq (MN_3 + \sqrt{8MN_3}) \sigma^2$, then the proof is completed.

   \end{IEEEproof}